\documentclass[11pt]{article}
\usepackage{amsfonts,amsthm,amssymb}
\usepackage[pdftex]{graphicx}
\usepackage{color}
\usepackage[fleqn]{amsmath}
\usepackage[noadjust]{cite}
\usepackage{hyperref}
\newcommand{\be}{\begin{eqnarray}}
\newcommand{\ee}{\end{eqnarray}}
\newcommand{\bez}{\begin{eqnarray*}}
\newcommand{\eez}{\end{eqnarray*}}

\newcommand{\bbC}{\mathbb{C}}

\newcommand{\bbR}{\mathbb{R}}

\newcommand{\imag}{\mathrm{i}}
\renewcommand{\Re}{\mathrm{Re}}
\renewcommand{\Im}{\mathrm{Im}}

\theoremstyle{plain}
\newtheorem{theorem}{Theorem}[section]

\newtheorem{proposition}[theorem]{Proposition}

\theoremstyle{definition}

\newtheorem{remark}[theorem]{Remark}
\newtheorem{example}[theorem]{Example}

\numberwithin{equation}{section}
\numberwithin{theorem}{section}

\renewcommand{\theequation} {\arabic{section}.\arabic{equation}}

\begin{document}

\title{\bf NLS breathers, rogue waves, and solutions \\ of the Lyapunov equation for Jordan blocks}

\author{
\sc{Oleksandr Chvartatskyi} and \sc{Folkert M\"uller-Hoissen} \\
 \small
 Max-Planck-Institute for Dynamics and Self-Organization, Am Fassberg 17, 37077 G\"ottingen, Germany \\
 E-mail: alex.chvartatskyy@gmail.com, folkert.mueller-hoissen@ds.mpg.de
}

\date{}

\maketitle

\begin{abstract}
The infinite families of Peregrine, Akhmediev and Kuznetsov-Ma breather solutions of the focusing Nonlinear Schr\"odinger 
(NLS) equation are obtained via a matrix version of the Darboux transformation, with a spectral matrix of the 
form of a Jordan block. The structure of these solutions is essentially determined by the corresponding 
solution of the Lyapunov equation. In particular, regularity follows from properties of the Lyapunov equation.
\end{abstract}

\section{Introduction}
\label{sec:intro}
In 1983 Howell Peregrine \cite{Pere83} showed that the ``self-focusing'' Nonlinear Schr\"odinger (NLS)
equation
\be
   \imag \, \mathfrak{q}_t + \mathfrak{q}_{xx} + 2 \, |\mathfrak{q}|^2 \, \mathfrak{q} = 0 
         \label{focNLS}
\ee
admits an exact solution that models a \emph{rogue wave} (or ``freak wave''), i.e., a wave with a significant 
amplitude that seems to ``appear from nowhere'' and ``disappears without a trace'' (cf. \cite{AAS09}). 
The existence of this \emph{Peregrine breather} is due to the modulational (Benjamin-Feir) instability of the 
background solution $e^{2 \imag \, t}$ on which it emerges. Up to this factor,  
it is given by a \emph{rational} expression, 
\bez
      \mathfrak{q} = \Big( 1 - \frac{G + \imag \, H}{D} \Big) \, e^{2 \imag t} \, ,   
\eez
where $D,G,H$ are real polynomials in the independent variables $x$ and $t$. 
That this ``instanton'' solution indeed models physical phenomena, has been demonstrated recently in 
nonlinear fibre optics \cite{KFFMDGAD10}, for deep water waves \cite{CHA11}, and in a multicomponent 
plasma \cite{BSN11}.\footnote{In some of these physical models, $x$ and $t$ as spatial and time variables 
are exchanged in (\ref{focNLS}). } Also see \cite{BKA09,Shri+Geog10}, for example.

The Peregrine breather is the analog of an NLS soliton, on the nontrivial background and for a special value 
of a spectral parameter. One can superpose an arbitrary number of such solutions in a nonlinear way. This includes 
an infinite family of exact solutions, where the $n$th member arises from a nonlinear superposition of $n$ 
Peregrine breathers, by taking a limit where associated spectral parameters tend to a special value. 
These solutions have the above quasi-rational form.  
The next to Peregrine member of this family apparently appeared first in \cite{AAT09}. 
It has also been shown to model certain waves observed in a water tank \cite{CHOA12}. 

To construct this family of solutions, the authors of \cite{AAS09}  
started from a linear system (Lax pair) for the focusing NLS equation, with a specific choice of the spectral parameter. 
Solutions are then constructed using Darboux transformations \cite{Matv+Sall91,Matv00}. The authors of \cite{AAS09} met  
the problem that the usual iteration of Darboux transformations requires a different spectral parameter at each step. 
But addressing quasi-rational solutions, one has to use the same spectral parameter, and this requires some 
modifications of the usual scheme. 
A corresponding improvement of 
the Darboux transformation method has been formulated in \cite{GLL12}. An even richer   
family of quasi-rational solutions than obtained in \cite{AAS09} has been provided in \cite{DGKM10,Duba+Matv11,Gail13}, 
expressing them as quotients of Wronskians. Via Hirota's bilinear method such solutions have been 
derived in \cite{Ohta+Yang12}.

Our present work uses yet another variant of the Darboux method. 
We use a \emph{single}, but \emph{vectorial} Darboux transformation and a linear system that involves a spectral 
parameter \emph{matrix} $Q$. This linear system is ``degenerate'' 
for a special eigenvalue of $Q$. The $n$th order quasi-rational solution of the focusing NLS equation is then 
obtained by choosing for $Q$ an $n \times n$ Jordan block with this special eigenvalue. 
This yields a nice characterization of the generalizations of the Peregrine breather.  
The vectorial Darboux transformation does not meet the problems mentioned in \cite{AAS09} and does not require 
modifications or generalizations as in \cite{GLL12}. 
 
Shortly after Peregrine's publication, a solution \cite{Akhm+Korn87} of the NLS equation was found which is also localized 
in ``time'' $t$.
But in contrast to the Peregrine breather, it is periodic in $x$-direction. 
This \emph{Akhmediev breather} yields the Peregrine breather in a certain limit. Also   
the Akhmediev breather belongs to an infinite family of exact solutions. 

Morover, there is a counterpart of the Akhmediev breather that is localized in $x$-direction and periodic in $t$-direction. 
It is known as \emph{Kuznetsov-Ma breather} \cite{Kuzn77,Ma79} and has been observed in fibre optics \cite{KFFMGWADD12}. 
The Akhmediev and Kuznetsov-Ma breathers can be combined into a common expression, from which we recover one or the other if  
the modulus of a spectral parameter is below or above a certain value. The Peregrine breather is then obtained for 
the transition value, the abovementioned special eigenvalue of $Q$. Both, the Akhmediev and the Kuznetsov-Ma breather, 
have meanwhile been observed in water tank experiments \cite{CKDA14}. Like the Peregrine breather, also the 
Akhmediev-Kuznetsov-Ma (AKM) breather belongs to an infinite sequence of exact solutions with increasing complexity. 

Whereas regularity of the generated solution is an easy problem in case of a single scalar Darboux transformation 
for the NLS equation, it is more involved in case of a vectorial Darboux transformation. But one can use  
properties of the Lyapunov equation (see the Appendices) to solve this problem satisfactorily. 
Moreover, the main purpose of this work is to show that the structure of the members of the families of Akhmediev, 
Kuznetsov-Ma and Peregrine breathers is largely determined by the solution of the rank one Lyapunov equation 
with a Jordan block matrix.

In Section~\ref{sec:DT} we recall a Darboux transformation result for the focusing NLS equation. This is then elaborated 
in Section~\ref{sec:exp_seed} for the abovementioned seed solution.  
Section~\ref{sec:concl} contains concluding remarks.

\section{Darboux transformation for the focusing NLS equation}
\label{sec:DT}
Let $\phi$ be a $2 \times 2$ matrix subject to the AKNS equation
\be
    \imag \, [J , \phi_{t}] + 2 \, \phi_{xx} - [[J,\phi],\phi_{x}] = 0 \, , \label{AKNS}
\ee
where $J = \mbox{diag}(1, -1)$. Decomposing $\phi$ as follows,
\be
    \phi = J \left(\begin{array}{cc} u & \mathfrak{q} \\ \mathfrak{r} & v \end{array}\right)
         = \left(\begin{array}{rr} u & \mathfrak{q} \\ - \mathfrak{r} & - v
         \end{array}\right) \, ,  
           \label{AKNS_decomp} 
\ee
and imposing the reduction condition
\be
        \phi^\dagger = \phi \, ,    \label{focNLS_red}
\ee
where $^\dagger$ indicates Hermitian conjugation, 
implies that $u$ and $v$ are real, $\mathfrak{r} = - \mathfrak{q}^\ast$ (where $^\ast$ denotes complex 
conjugation), and (\ref{AKNS}) becomes 
\be
   && \imag \, \mathfrak{q}_t + \mathfrak{q}_{xx} + \mathfrak{q} \, v_x + u_x \, \mathfrak{q} = 0 \, , \qquad
      (u_x - |\mathfrak{q}|^2)_x = 0 \, , \quad (v_x - |\mathfrak{q}|^2)_x = 0 \, .   \label{prefocNLS}
\ee
Setting possible constants of integration (actually functions of $t$) to zero by integrating the last 
two equations, and eliminating $u_x$ and $v_x$ in the first equation, we obtain the focusing NLS 
equation (\ref{focNLS}). Any solution of the a priori more general system (\ref{prefocNLS}) 
can be transformed into a solution of the focusing NLS equation (cf. \cite{DMH10NLS}, for example). 

The following Darboux transformation is obtained from a general binary Darboux transformation result 
\cite{DMH13SIGMA} in bidifferential calculus, see \cite{CDMH16}.

\begin{proposition}
\label{prop:NLS}
Let $\phi_0$ be a solution of (\ref{AKNS}) and (\ref{focNLS_red}), 
$Q$ a constant $n \times n$ matrix satisfying the spectrum condition
\be
    \sigma(Q) \cap \sigma(-Q^\dagger) = \emptyset \, ,     \label{spec}
\ee  
and $\eta$ an $n \times 2$ matrix solution of the linear 
system
\be
    \imag \, \eta_t = - Q \, \eta_x + \eta \,  \phi_{0,x} \, , \qquad
    \eta_x = - \frac{1}{2} Q \, \eta J - \frac{1}{2} \eta \, [J,\phi_0] \, .   \label{NLS_lin_eqs_theta} 
\ee
Furthermore, let $\Omega$ be an $n \times n$ matrix solution of the Lyapunov equation
\be
       Q \, \Omega + \Omega \, Q^\dagger = \eta \, \eta^\dagger \, .    \label{NLS_Lyapunov} 
\ee
Then, wherever $\Omega$ is invertible,  
\be
    \phi = \phi_0 - \eta^\dagger \, \Omega^{-1} \, \eta     \label{NLS_phi_sol}
\ee
is a new solution of (\ref{AKNS}) and (\ref{focNLS_red}). As a consequence,
the components of $\phi$ yield a solution of (\ref{prefocNLS}).  \hfill $\Box$
\end{proposition}

(\ref{NLS_lin_eqs_theta}) is a Lax pair for (\ref{AKNS}), i.e., (\ref{AKNS}) is the compatibility 
condition arising from $\eta_{xt} = \eta_{tx}$. Note that $Q$ plays the role of a \emph{matrix} 
spectral parameter. 

Solutions of (\ref{AKNS}) obtained via Proposition~\ref{prop:NLS} in this work indeed yield via (\ref{AKNS_decomp}) 
directly solutions $\mathfrak{q}$ of the focusing NLS equation (\ref{focNLS}). This is so because we start from 
a seed solution that satisfies the focusing NLS equation and which determines the asymptotics of the generated solutions. 

\begin{remark}
Dropping the spectrum condition (\ref{spec}) in Proposition~\ref{prop:NLS}, we need to add 
the following two equations
\bez
   \Omega_x = - \frac{1}{2} \eta \, J \, \eta^\dagger \, ,  \qquad
    \imag \, \Omega_t = \frac{1}{2} ( Q \eta \, J \, \eta^\dagger - \eta J \eta^\dagger Q^\dagger )
                     + \frac{1}{2} \eta \, [J , \phi_0 ] \, \eta^\dagger  \, ,  
\eez
which otherwise follow by differentiation of (\ref{NLS_Lyapunov}) and application of the remaining equations. 
Using these equations, (\ref{NLS_lin_eqs_theta}),  
(\ref{NLS_Lyapunov}) and the assumption that $\phi_0$ solves (\ref{AKNS}), one can also prove the 
proposition by a direct but lengthy computation.  \hfill $\Box$
\end{remark}

\begin{remark}
The (algebraic) Lyapunov equation plays an important role in the stability analysis of systems of ordinary differential 
equations (see, e.g., \cite{Lanc+Tism85}). It is a special case of Sylvester's equation \cite{Horn+John13}. 
An appearance of the latter in a Riemann-Hilbert factorization problem dates back to 1986 \cite{Sakh86}.
Sylvester's equation typically enters the stage when matrix (or operator) versions of integrable (differential or difference) equations 
are considered \cite{March88}, or when matrix methods are applied, e.g., to concisely express iterated B\"acklund or 
Darboux transformations (see, e.g., \cite{SSR13} and references therein, and \cite{ABDM10,XZZ14}). 
It is therefore not surprising that, for various integrable equations, different specializations of 
Sylvester's equation show up via the universal binary Darboux transformation in the framework 
of bidifferential calculus (see \cite{DMH13SIGMA,CDMH16} and references cited there).   
If the ``input matrices'' of the Sylvester equation (in our particular case $Q$ and $Q^\dagger$) 
are diagonal, the solution is a Cauchy-like matrix. The generalization to non-diagonalizable input matrices, hence going 
beyond a ``Cauchy matrix approach'' (see, e.g., \cite{NAH09,Schie10LAA}), reaches solutions that are otherwise 
obtained via higher-order poles in the inverse scattering method or special limits of multi-soliton solutions. 
This generalization is of uttermost importance for our present work. \hfill $\Box$
\end{remark}

\section{Exact solutions obtained from a simple seed solution}
\label{sec:exp_seed}
A simple solution of the focusing NLS equation is given by $\mathfrak{q}_0 = e^{2 \imag t}$.\footnote{Applying 
$t\mapsto \alpha^2 t$ and $x \mapsto \alpha \, x$, with a real constant $\alpha >0$, to any solution $\mathfrak{q}$ 
of the NLS equation, and multiplying the resulting expression by $\alpha$, yields a new solution. In this way 
we obtain in particular the more general solution $ \alpha \, e^{ 2 \imag \alpha^2 t } $. }
This solution is unstable (Benjamin-Feir instability), which may be regarded as the origin of the occurrence of rogue waves. 
A solution of (\ref{AKNS}) and (\ref{focNLS_red}), corresponding to $\mathfrak{q}_0$, is given by 
\bez
     \phi_0 = \left(\begin{array}{cc} x & e^{2 \imag t}  \\ 
                                      e^{-2 \imag t} & - x \end{array}\right) \, .
\eez
We note that $\phi_{0,x} = J$. Writing 
\bez
  \eta = \left( \eta_1 \, , \, \eta_2 \right) \, e^{-\imag J t} 
         = \left( \eta_1 \, e^{-\imag t} \, , \,  \eta_2 \, e^{\imag t} \right) \, ,
\eez
where $\eta_1$ and $\eta_2$ are $n$-component vectors, 
the linear system (\ref{NLS_lin_eqs_theta}) takes the form
\bez
  &&  \imag \, \eta_{j,t} + Q \, \eta_{j,x} = 0 \qquad j=1,2 \, ,  \\
  &&  \eta_{1,x} + \frac{1}{2} Q \, \eta_1 - \eta_2 = 0 \, , \quad
      \eta_{2,x} - \frac{1}{2} Q \, \eta_2 + \eta_1 = 0 \, ,
\eez
which decouples to
\be
     \eta_{1,xx} - (\frac{1}{4} Q^2-I) \, \eta_1 = 0 \, , \qquad
     \imag \, \eta_{1,t} + Q \, \eta_{1,x} = 0 \, , \qquad
     \eta_2 = \frac{1}{2} Q \, \eta_1 + \eta_{1,x}   \, .     \label{eta_lin_sys}
\ee
Given a solution of the first two equations, we have to find a corresponding solution $\Omega$ of (\ref{NLS_Lyapunov}), 
i.e., the rank two Lyapunov equation
\be
       Q \Omega + \Omega Q^\dagger = \eta_1 \, \eta_1^\dagger + \eta_2 \, \eta_2^\dagger \, . \label{NLS_Lyapunov_2}
\ee
The solution can be written\footnote{As a consequence of (\ref{spec}), the Lyapunov equation possesses a 
unique solution.}
as $\Omega = \Omega_1 + \Omega_2$, where $\Omega_1$ and $\Omega_2$ are solutions 
of the rank one Lyapunov equations
\be
       Q \Omega_1 + \Omega_1 Q^\dagger = \eta_1 \, \eta_1^\dagger \, , \qquad
       Q \Omega_2 + \Omega_2 Q^\dagger = \eta_2 \, \eta_2^\dagger \, .     \label{NLS_Lyapunov_3}
\ee
According to Proposition~\ref{prop:NLS}, 
\be
    \mathfrak{q} = \left( 1 - \eta_1^{\dagger} \, \Omega^{-1} \, \eta_2 \right) \, e^{2\,\imag\,t}    \label{NLS_q_sol}
\ee
is then a solution of the focusing NLS equation.

Let $Q$ now be a lower triangular Jordan block, 
\be
   Q = \left( \begin{array}{ccccc} q  & 0      & \cdots & \cdots & 0      \\ 
                                       1  & q  & \ddots  & \ddots & 0      \\ 
                                       0  & \ddots & \ddots & \ddots & \vdots \\
                                   \vdots & \ddots & \ddots & \ddots & 0      \\
                                       0  & \cdots &  0  &   1    & q 
                 \end{array} \right)  \, ,   \label{Q_Jordan_block}
\ee
the case we are concentrating on in this work. The solution of 
the corresponding rank one Lyapunov equation is then given explicitly in Proposition~\ref{prop:rank_one} in Appendix~A. 
We will use an index $(k)$ to specify the matrix size. For example, $Q_{(k)}$ is the $k \times k$ 
matrix version of $Q$ and $\eta_{i (k)} = (\eta_{i1}, \eta_{i2},\ldots,\eta_{ik})^\intercal$. Furthermore, we set 
\bez
     \kappa := 2 \, \Re(q) \, , \qquad 
     \tilde{\eta}_{i,k+1} := \eta_{i,k+1} - \kappa^{-1} \eta_{ik} \, , \quad k=1,2, \ldots \, .
\eez

\begin{example}
\label{ex:n=2Jordan_Lyapunov_sol}
 For $n=2$, using Proposition~\ref{prop:rank_one} we obtain
\bez
  \Omega_{(2)} = \Omega_{(2)1} + \Omega_{(2)2}
               = \frac{1}{ \kappa } \sum_{i=1}^2 \left( \begin{array}{cc}
                |\eta_{i1}|^2 & \eta_{i1} \tilde{\eta}_{i2}^\ast  \\
                 \eta_{i1}^\ast \tilde{\eta}_{i2} &
                    | \tilde{\eta}_{i2} |^2 + \kappa^{-2} |\eta_{i1}|^2 
                     \end{array} \right)  \, .
\eez
The determinant of $\Omega_{(2)}$  can be written in the following concise form, 
\bez
    \det(\Omega_{(2)}) = \kappa^{-4} (|\eta_{11}|^2 + |\eta_{21}|^2)^2 
            + \kappa^{-2} \, |\det(\eta_1,\eta_2)|^2  \, ,
\eez
where $\det(\eta_1,\eta_2) = \eta_{11} \eta_{22} - \eta_{12} \eta_{21}$. 
Hence
\be
     \Omega_{(2)}^{-1} = \frac{ \kappa }{\det(\Omega)} \sum_{i=1}^2 \left( \begin{array}{cc}
               | \tilde{\eta}_{i2} |^2 + \kappa^{-2} |\eta_{i1}|^2  & - \eta_{i1} \tilde{\eta}_{i2}^\ast  \\
                - \eta_{i1}^\ast \tilde{\eta}_{i2} & |\eta_{i1}|^2
                     \end{array} \right) \, .         \label{invOm_n=2}
\ee
This leads to the following solution of the focusing NLS equation, 
\be
    \mathfrak{q} = \left( 1 - \frac{F}{\det(\Omega_{(2)})} \right) \, e^{2\,\imag\,t}   \, ,  \label{n=2sol}
\ee
where
\bez
    F = \frac{1}{4 \Re(q)^3} \Big( \eta_{21} \eta_{11}^* \, ( |\eta_{11}|^2 + |\eta_{21}|^2 ) 
      + 2 \imag \, \Re(q) \, \Im[ (\eta_{11}^*)^2  \det(\eta_1,\eta_2)] \Big) \, ,
\eez  
and $\eta_1,\eta_2$ have to solve the linear system (\ref{eta_lin_sys}). The solutions of the 
latter will be provided below. 

 \hspace{1cm} \hfill $\Box$
\end{example}

 From the Lyapunov equation (\ref{NLS_Lyapunov_2}) one finds that its solutions are nested:
\bez
     \Omega_{(n+1)} = \left( \begin{array}{cc}
                \Omega_{(n)} & B_{(n+1)}  \\
                 B_{(n+1)}^\dagger & \omega_{(n+1)} \end{array} \right)  \, ,                    
\eez
where
\bez
  && B_{(n+1)} = K^{-1} \Big( \eta_{1(n)} \, \eta_{1,n+1}^\ast + \eta_{2(n)} \, \eta_{2,n+1}^\ast
                 - \Omega_{(n)} (0,\ldots,0,1)^\intercal \Big) \, , \\
  && \omega_{(n+1)} = \frac{1}{\kappa} \Big( |\eta_{1,n+1}|^2 + |\eta_{2,n+1}|^2 
                       - 2 \, \Re[(0,\ldots,0,1) B_{(n+1)}] \Big) \, ,
\eez
and $K$ is the Jordan block $Q_{(n)}$ with $q$ replaced by $\kappa$. 
In order to evaluate (\ref{NLS_q_sol}), we need the inverse of $\Omega_{(n+1)}$. 
Using a well-known formula, it is given by 
\be
   \Omega_{(n+1)}^{-1} = \left( \begin{array}{cc}
      \Omega_{(n)}^{-1} - S_{\Omega_{(n)}}^{-1} \Omega_{(n)}^{-1} B_{(n+1)} B_{(n+1)}^\dagger \Omega_{(n)}^{-1}
             & - S_{\Omega_{(n)}}^{-1} \Omega_{(n)}^{-1} B_{(n+1)}  \\
      - S_{\Omega_{(n)}}^{-1} B_{(n+1)}^\dagger \Omega_{(n)}^{-1} &  S_{\Omega_{(n)}}^{-1} 
       \end{array} \right) \, ,   \label{Omega_inverse}
\ee
with the scalar Schur complement 
\bez
    S_{\Omega_{(n)}} = \omega_{(n+1)} - B_{(n+1)}^\dagger \Omega_{(n)}^{-1} B_{(n+1)} \, .
\eez
If $\Omega_{(n)}$ and $\Omega_{(n+1)}$ are invertible, then also $S_{\Omega_{(n)}}$. The above equations allow  
to recursively compute the solutions of the Lyapunov equation (\ref{NLS_Lyapunov_2}) for $n=3,4,\ldots$,  
their inverses, and the corresponding NLS solution (\ref{NLS_q_sol}). The concrete expressions for the solutions 
$\eta_1,\eta_2$ of the linear system (\ref{eta_lin_sys}) will be provided in following subsections.

\begin{example}
\label{ex:n=3Jordan_Lyapunov_sol}
 For $\Omega_{(3)}$ we obtain  
\bez
 &&  B_{(3)} = \frac{1}{ \kappa } \sum_{i=1}^2  \left( \begin{array}{c}   \eta_{i1} \, \tilde{\eta}_{i3}^\ast \\
                \tilde{\eta}_{i2} \, \tilde{\eta}_{i3}^\ast + \kappa^{-2} \eta_{i1} \, \tilde{\eta}_{i2}^\ast 
                          - \kappa^{-3} |\eta_{i1}|^2 \end{array} \right) \, , \\
 && \omega_{(3)} = \frac{1}{ \kappa } \sum_{i=1}^2 \Big( | \tilde{\eta}_{i3} |^2 
          + \kappa^{-2} | \tilde{\eta}_{i2} - \kappa^{-1} \eta_{i1} |^2 + \kappa^{-4} |\eta_{i1}|^2 \Big) \, .                
\eez
Since $\Omega_{(2)}$ and its inverse (see (\ref{invOm_n=2})) have been elaborated in the preceding example,  
we can now use (\ref{Omega_inverse}) to compute the inverse of $\Omega_{(3)}$ and then 
elaborate (\ref{NLS_q_sol}). 
\hfill $\Box$
\end{example}

\begin{remark}
\label{rem:invariance}
Let $Q$ be an $n \times n$ lower triangular Jordan block. Any invertible lower triangular $n \times n$ Toeplitz 
matrix $A$ with constant entries commutes with it. 
As a consequence, (\ref{NLS_q_sol}) is invariant under
\bez
     \eta_1 \mapsto A \, \eta_1 \, , \qquad  \eta_2 \mapsto A \, \eta_2 \, ,
\eez
since this is a symmetry of (\ref{eta_lin_sys}), and from (\ref{NLS_Lyapunov_2}) we find that 
$\Omega \mapsto A \Omega A^\dagger$.    \hfill $\Box$
\end{remark}

\begin{remark}
By use of the matrix determinant lemma, (\ref{NLS_q_sol}) can also be expressed as
\bez
   \mathfrak{q} = \frac{\det(\Omega_{(n)} - \eta_2 \, \eta_1^\dagger)}{\det(\Omega_{(n)})} \,  e^{2\,\imag\,t} \, .
\eez
It follows from Remark~\ref{rem:invariance} that, if $Q$ is an $n \times n$ lower triangular Jordan block, then 
under a transformation of $\eta_1$ and $\eta_2$ with  
an invertible lower triangular Toeplitz matrix $A$ we have $\det(\Omega_{(n)}) \mapsto |a_1|^2 \det(\Omega_{(n)})$ and 
$\det(\Omega_{(n)} - \eta_2 \, \eta_1^\dagger) \mapsto |a_1|^2 \det(\Omega_{(n)} - \eta_2 \, \eta_1^\dagger)$, where 
$a_1$ is the eigenvalue of $A$.
\hfill $\Box$
\end{remark}

\subsection{The infinite family of Akhmediev and Kuznetsov-Ma breathers}
If $\frac{1}{4} Q^2-I$ is invertible, then it possesses a square root\footnote{Another square root is $-\Lambda$, 
of course. We note that a matrix may possess more than two square roots (see, e.g., \cite{JOR01,Hasa97}), but this 
is not of relevance here. } 
$\Lambda$ \cite{Cros+Lanc74}, and the linear system (\ref{eta_lin_sys}) admits the following solutions,
\be
    \eta_1 = \cosh(\Theta) \, a  \, , \qquad
    \eta_2 = \frac{1}{2} Q \, \eta_1 + \Lambda \, \sinh(\Theta) \, a   \, ,   \label{eta_sol}
\ee
with a constant $n$-component vector $a$ and 
\bez
     \Theta := \Lambda \, (x + \imag \, Q t) + C \, ,
\eez
where $C$ is a constant $n \times n$ matrix that commutes with $\Lambda$ (and thus with $Q$). 

\begin{remark}
The solution of (\ref{NLS_Lyapunov}), with $\eta$ given by (\ref{eta_sol}), can also be expressed as
\bez
 \Omega_{(n)} = \cosh(\Theta) X_{11} \cosh(\Theta^\dagger) 
          + \cosh(\Theta) X_{12} \sinh(\Theta^\dagger) + \sinh(\Theta) X_{21} \cosh(\Theta^\dagger) 
          + \sinh(\Theta) X_{22} \sinh(\Theta^\dagger) \, ,
\eez
where the four constant $n \times n$ matrices $X_{ij}$ have to satisfy the Lyapunov equations
\bez
  && Q X_{11} + X_{11} \, Q^\dagger = a a^\dagger + \frac{1}{4} \, Q \, a a^\dagger \, Q^\dagger  \, , \qquad
     Q X_{22} + X_{22} \, Q^\dagger = \Lambda \, a a^\dagger \, \Lambda^\dagger \, , \\
  && Q X_{12} + X_{12} \, Q^\dagger = \frac{1}{2} \, Q \, a a^\dagger \, \Lambda^\dagger  \, , \qquad
     Q X_{21} + X_{21} \, Q^\dagger = \frac{\imag}{2} \, \Lambda \, a a^\dagger \, Q^\dagger   \, .
\eez
We observe that, if $X$ solves the rank one Lyapunov equation
\be
       Q X + X Q^\dagger = a a^\dagger \, ,   \label{Lyapunov_rank_one}
\ee
then $X_{11} = X + \frac{1}{4} \, Q \, X \, Q^\dagger$, 
$X_{12} = X_{21}^\dagger = \frac{1}{2} \, Q \, X \, \Lambda^\dagger$, 
$X_{22} = \Lambda \, X \, \Lambda^\dagger$, solve the preceding system. 
\hfill $\Box$
\end{remark}

\subsubsection{Single Akhmediev and Kuznetsov-Ma breathers}
\label{subsec:AKM_breather}
In the simplest case, $n=1$, we write $q$ instead of $Q$, $c$ instead of $C$, and $\lambda$ instead of $\Lambda$. 
 From (\ref{NLS_Lyapunov_2}) we obtain $\Omega_{(1)} = (2 \, \Re(q))^{-1} (|\eta_1|^2 + |\eta_2|^2)$,
which shows that $\Omega_{(1)}$ nowhere vanishes if $a \neq 0$. We obtain the following regular 
solution of the focusing NLS equation,
\be
  \mathfrak{q} =  (1 - \frac{1}{\Omega_{(1)}} \, \eta_1^\ast \, \eta_2 ) \, e^{2\imag t} 
               = \Big( 1 
    - 4 \Re(q) \, \frac{ q \, |\cosh (\vartheta)|^2 + 2 \lambda \cosh (\vartheta)^* \, \sinh (\vartheta) }
           {4 \left| \cosh (\vartheta)\right|^2 + |q \cosh (\vartheta) + 2 \lambda  \sinh (\vartheta)|^2} 
                              \Big) \, e^{2\imag t}  \, ,       \label{n=1_breather}           
\ee
with $\lambda = \frac{1}{2} \sqrt{q^2-4}$ and 
\bez
      \vartheta := \lambda \, ( x + \imag \, q \, t ) + c 
                 &=& \Re(\lambda) \, x - [ \Re(\lambda) \, \Im(q) + \Im(\lambda) \, \Re(q) ] \, t  + \Re(c) \\
                 &&  + \imag \, \Big( \Im(\lambda) \, x + [ \Re(\lambda) \, \Re(q) - \Im(\lambda) \, \Im(q) ] \, t + \Im(c) \Big) \, .
\eez
As is evident from (\ref{NLS_Lyapunov}) and (\ref{NLS_phi_sol}), the constant $a$ drops out.

If $q \neq 0$ is real and $|q|>2$, then $\lambda$ is real and $\vartheta = \lambda \, [ x-x_0 + \imag \, q \, (t-t_0) ]$, where 
$x_0 = -\Re(c)/\lambda$ and $t_0 = - \Im(c)/(\lambda q)$.
The solution then becomes the \emph{Kuznetsov-Ma breather} \cite{Kuzn77,Ma79}, 
\bez
  \mathfrak{q} = - \Big(   1 + 2 \lambda \, \frac{ 2 \lambda \, \cos(\vartheta_1) + \imag \, q \, \sin(\vartheta_1)}
        {2 \, \cos(\vartheta_1) + 2 (1+\lambda^2) \, \cosh(\vartheta_2) + \lambda q \, \sinh(\vartheta_2) } 
            \Big) \, e^{2\imag t}   \, ,
\eez
where $\vartheta_1 = 2 \lambda q \, (t-t_0)$ and 
$\vartheta_2 = 2 \lambda \, (x-x_0)$. See Fig.~\ref{fig:n=1_breathers} for a plot.

\begin{figure} 
\begin{center}
\includegraphics[scale=.5]{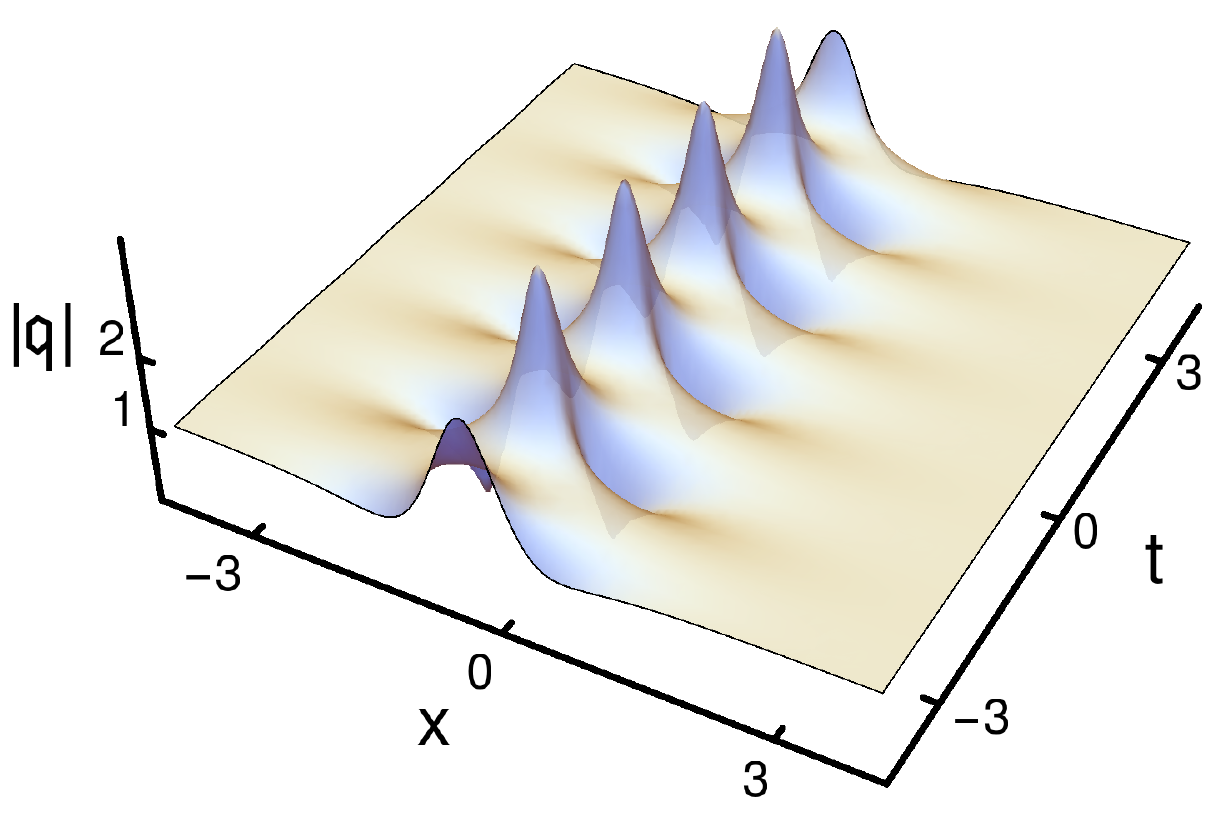} 
\hspace{.5cm}
\includegraphics[scale=.5]{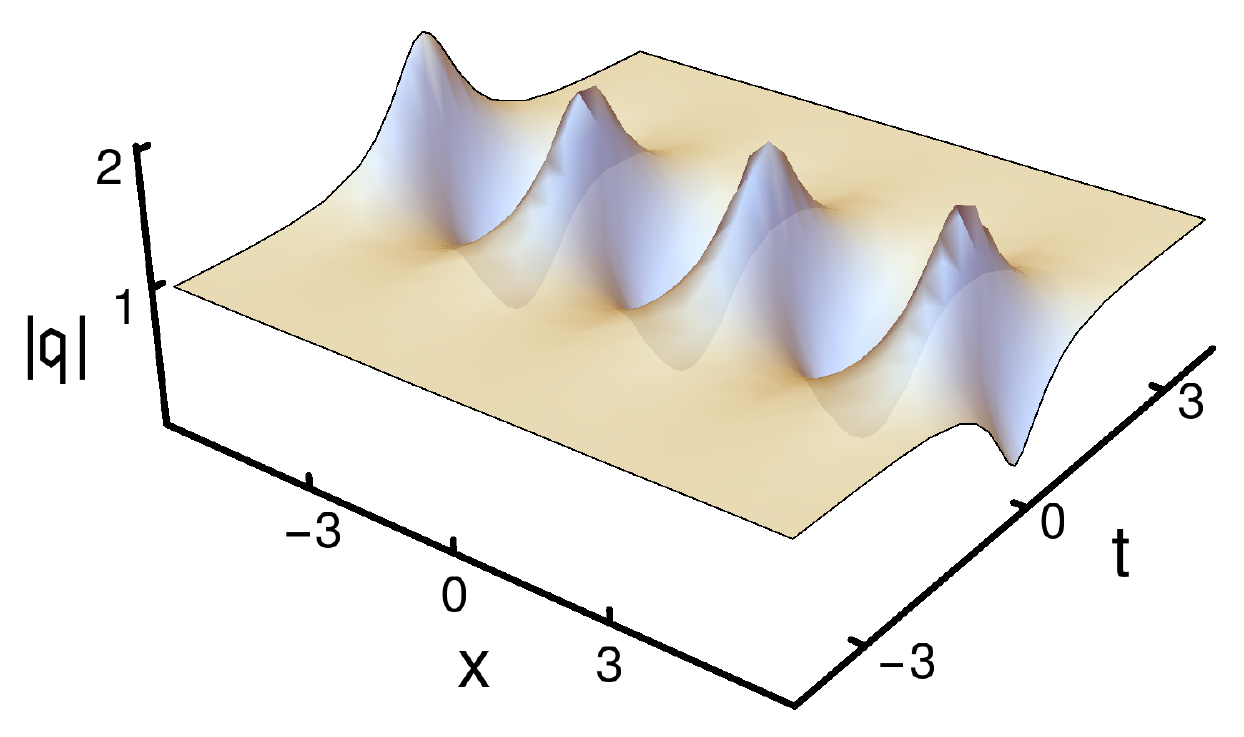} 
\end{center}
\caption{Plot of the modulus of the Kuznetsov-Ma breather (left plot, with $q=5/2$ and $c=0$) and the Akhmediev breather (right plot, 
with $q=1$ and $c=0$) in Section~\ref{subsec:AKM_breather}. }
\label{fig:n=1_breathers} 
\end{figure}

If $q \neq 0$ is real and $|q| < 2$, then $\lambda$ is imaginary. In this case we write 
$\lambda = \imag \, \tilde{\lambda}$ with real $\tilde{\lambda}$, so that $q^2 = 4 (1-\tilde{\lambda}^2)$.
Then the above solution becomes the \emph{Akhmediev breather}\footnote{Interpreting $t$ as time, this is actually 
a rogue wave which suddenly appears and then disappears. Changing the roles of $x$ and $t$ as spatial and temporal coordinates, 
it becomes a breather.}
\cite{AEK87}, 
\bez
    \mathfrak{q} = \Big( 4 \tilde{\lambda} \, \frac{ 2 \tilde{\lambda} \, \cosh(\tilde{\vartheta}_1)
                   + \imag  \, q \, \sinh(\tilde{\vartheta}_1)}{4 \cosh(\tilde{\vartheta}_1) 
                   + q^2 \cos(\tilde{\vartheta}_2) - 2 \tilde{\lambda} q \, \sin(\tilde{\vartheta}_2)} - 1 \Big) 
                   \, e^{2\imag t}   \, ,   
\eez
with $\tilde{\vartheta}_1 = 2 \tilde{\lambda} q \, (t - t_0 )$, $t_0 = \Re(c)/(\tilde{\lambda} q)$, and 
$\tilde{\vartheta}_2 = 2 \tilde{\lambda} \, (x-x_0)$, $x_0 = - \Im(c)/\tilde{\lambda}$. 
See Fig.~\ref{fig:n=1_breathers} for a plot.

Both, the Akhmediev and the Kuznetsov-Ma breathers, are thus special cases of the solution (\ref{n=1_breather}). The latter also 
contains analogs of the Kuznetsov-Ma breather with non-zero constant velocity, see Fig.~\ref{fig:moving_n=1_breather} and 
also \cite{Bion+Kova14}. 
These solutions are \emph{not} obtained via a Galilean transformation from those with real $q$ (cf. \cite{ASA09PRA}, 
for example). Generalizations of above solutions involving Jacobi elliptic functions appeared in \cite{AEK87}.

\begin{figure} 
\begin{center}
\includegraphics[scale=.5]{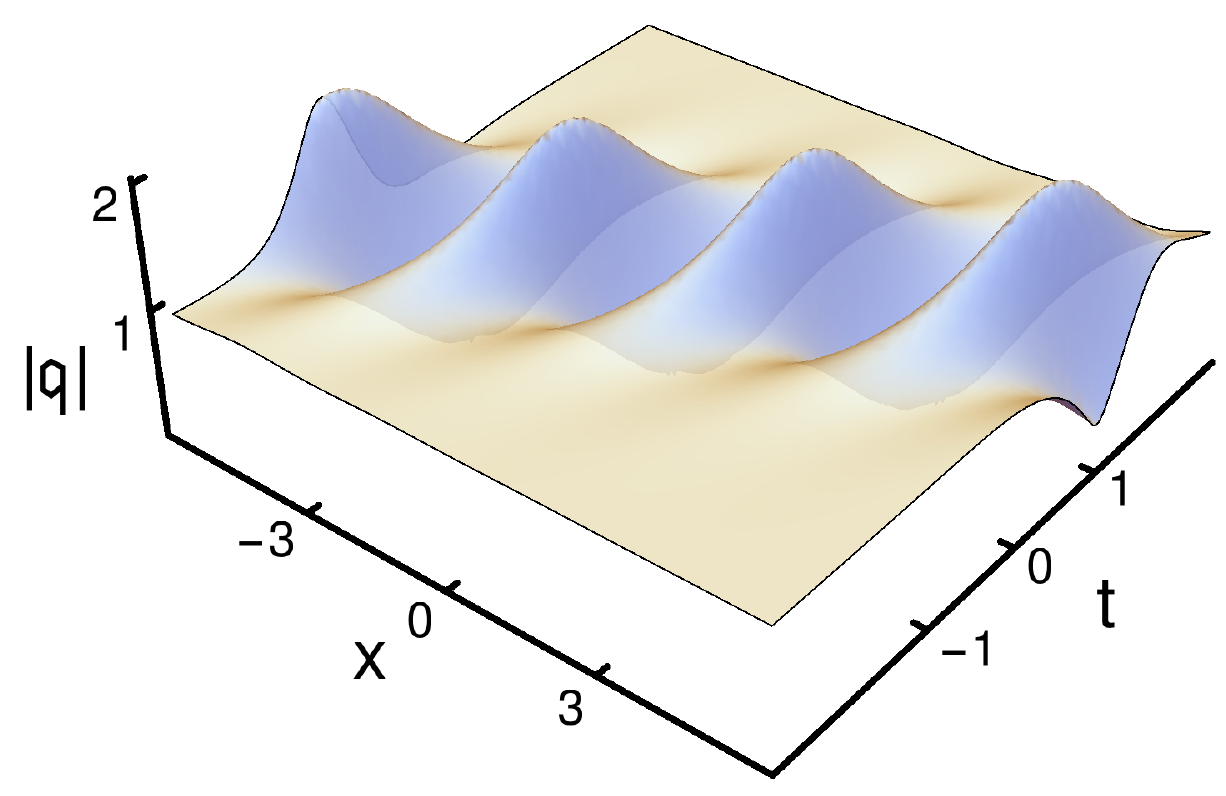}
\end{center}
\caption{Plot of the modulus of the solution (\ref{n=1_breather}) in Section~\ref{subsec:AKM_breather} 
with $c=0$ and $q=1+\imag$. }
\label{fig:moving_n=1_breather} 
\end{figure}
\hfill $\Box$

\subsubsection{Higher order Akhmediev-Kuznetsov-Ma breathers}
\label{subsec:higherAKM}
Let us split the $n \times n$ lower triangular Jordan block (\ref{Q_Jordan_block}) as
\bez
   Q = q \, I_n + N \, ,     
\eez
where $I_n$ denotes the $n \times n$ identity matrix and $N$ is the nilpotent matrix of degree $n$, 
having $1$'s in the second lower diagonal and otherwise zeros. 
For any analytic function $f$ of a single variable, we define 
\bez
     f(Q) := \sum_{k=0}^{n-1} \frac{1}{k!} f^{(k)}(q) \, N^k \, ,
\eez
using the Taylor series expansion for $f$. 
If $q \neq \pm 2$, then $\frac{1}{4} Q^2-I$ is invertible and has the square root (also see \cite{Higham87,High08})
\bez
    \Lambda = f(Q)  \qquad \mbox{where} \quad   f(q) = \frac{1}{2} \, \sqrt{q^2 - 4}  \, .    
\eez
This is the $n \times n$ lower triangular Toeplitz matrix
\bez
    \Lambda = \frac{1}{2} \, \sqrt{q^2 - 4} \, \left( \begin{array}{cccccc} 1 & 0 & \cdots & \cdots & \cdots & 0 \\ 
                                      q \, (q^2 - 4)^{-1}  &  1  & \ddots & \ddots & \ddots & \vdots \\
                                     - 2 \, (q^2 - 4)^{-2} & \ddots  &  \ddots & \ddots & \ddots & \vdots \\
                              2 \, q \, (q^2-4)^{-3} & \ddots & \ddots  &  \ddots &  \ddots & \vdots \\
                                 \vdots & \ddots & \ddots & \ddots & \ddots & 0 \\
                                 \vdots &  & \ddots & \ddots & \ddots & 1
                     \end{array} \right) \, .   
\eez

\begin{proposition}
\label{prop:regularity}
Let $Q$ be given by (\ref{Q_Jordan_block}) with $\Re(q) \neq 0$ and $q \neq \pm 2$.   
Then the solution (\ref{NLS_q_sol}) of the focusing NLS equation, obtained from the solution 
of the linear system determined by (\ref{eta_sol}), is regular if the first component of the vector $a$ 
in (\ref{eta_sol}) is different from zero.  
\end{proposition}
\begin{proof}
$q \neq \pm 2$ is the condition for $\frac{1}{4} Q^2 -I$ to be invertible.
As a consequence of the quadratic identity for $\sinh$ and $\cosh$ (respectively, $\sin$ and $\cos$), which is also satisfied 
by their matrix versions \cite{High08}, $\eta_1$ and $\eta_2$ given by (\ref{eta_sol}) cannot vanish simultaneously at any 
space-time point $(x,t)$ if $a_1 \neq 0$. Now it follows from Proposition~\ref{prop:Lyapunov_gen} in Appendix~B that 
the solution of the Lyapunov equation (\ref{NLS_Lyapunov_2}) is invertible 
for all $x \in \bbR$ and $t \in \bbR$. But this is the regularity condition for (\ref{NLS_q_sol}). 
\end{proof}

We define the \emph{$n$-th order AKM (Akhmediev-Kuznetsov-Ma) breather} as the solution 
of the focusing NLS equation, obtained via Proposition~\ref{prop:NLS} 
with the $n \times n$ matrix $Q$ given by the Jordan block (\ref{Q_Jordan_block}), where 
$\Re(q) \neq 0$ and $q \neq \pm 2$, 
and with the solution (\ref{eta_sol}) of the linear system, where the first component of the vector $a$ 
is different from zero. According to Proposition~\ref{prop:regularity}, the solution is regular.
The first component of $a$ can be set to $1$ since it does not influence the solution determined 
by (\ref{NLS_q_sol}). Writing $a = A \, ( 1 , 0, \ldots, 0 )^\intercal$ (with the Toeplitz matrix $A$ 
in (\ref{prop_A,D})), Remark~\ref{rem:invariance} shows that the parameters $a_i$, $i>1$, are redundant. 
Hence, without restriction of generality we can set 
\bez
        a = (1,0,\ldots,0) \, .
\eez
The $n$-th order AKM breather thus depends on $n$ complex parameters. 
The parameter $c_1$ can be removed by shifts in $x$ and $t$, which are Lie point symmetries of the NLS equation.

\begin{example}
\label{ex:n=2breather}
 For $n=2$ we have
\bez
  &&  Q = \left( \begin{array}{cc} q & 0 \\ 
                1 & q  
                \end{array} \right) \, ,    \qquad
    \Lambda = \left( \begin{array}{cc} \lambda & 0 \\ 
                \frac{q}{4 \lambda} & \lambda 
                \end{array} \right) \, ,    \qquad
    \Theta = \left( \begin{array}{cc} \vartheta & 0 \\ 
                \frac{1}{4 \lambda} \mathcal{P} & \vartheta 
                \end{array} \right) \, ,    \\
  &&  \eta_1 = \left( \begin{array}{c} \cosh(\vartheta) \\ 
                \frac{1}{4 \lambda} \mathcal{P} \, \sinh(\vartheta) 
                \end{array} \right) \, , \qquad
     \eta_2 = \left( \begin{array}{c} \theta_1 \\ 
                \frac{\mathcal{P}+2}{4 \lambda} \, \theta_2  
                \end{array} \right)  \, ,
\eez
where $\lambda = \frac{1}{2} \, \sqrt{q^2 - 4}$,  
\bez
  &&  \vartheta = \lambda \, (x + \imag \, q \, t) + c_1 \, , \qquad
    \mathcal{P} = q \, x + 2 \imag \, (q^2 - 2) \, t + 4 \lambda \, c_2 \, , \\
  && \theta_1 = \frac{q}{2} \cosh(\vartheta) + \lambda \, \sinh(\vartheta) \, , \qquad
     \theta_2 = \lambda \, \cosh(\vartheta) + \frac{q}{2} \, \sinh(\vartheta) \, ,
\eez
and $c_1,c_2$ are arbitrary complex constants. 
The determinant of $\Omega_{(2)}$ is
\bez
   \det(\Omega_{(2)}) &=& \frac{1}{(2 \, \Re(q))^4} \left( |\cosh(\vartheta)|^2 + |\theta_1|^2 \right)^2
    + \frac{1}{4 \, \Re(q)^2} \Big| \frac{\mathcal{P}+2}{4 \lambda} \, \cosh(\vartheta) \, \theta_2
        - \frac{\mathcal{P}}{4 \lambda} \, \sinh(\vartheta) \, \theta_1 \Big|^2   \\
      &=& \frac{1}{(2 \, \Re(q))^4} \left( | \cosh(\vartheta) |^2 
              + | \frac{q}{2} \cosh(\vartheta) + \lambda \, \sinh(\vartheta) |^2 \right)^2 \\
      && + \frac{1}{16 \, \Re(q)^2} \Big| \mathcal{P} + 2 \, \cosh(\vartheta)^2  
                + \frac{q}{\lambda} \, \cosh(\vartheta) \, \sinh(\vartheta) \Big|^2 \, .
\eez
The $n=2$ AKM solution of the focusing NLS equation is then given by (\ref{n=2sol}) with
\bez
     F 
    &=& \frac{1}{4 \, \Re(q)^3}  
         \Big[ \left( |\cosh(\vartheta)|^2 + |\frac{q}{2} \cosh(\vartheta) + \lambda \, \sinh(\vartheta)|^2 \right) 
                 \, \cosh(\vartheta) \, (\frac{q}{2} \cosh(\vartheta) + \lambda \, \sinh(\vartheta) )  \\
     && + \frac{\imag}{2} \, \Re(q) \, \Im\Big( \mathcal{P} + 2 \, \cosh(\vartheta)^2  
                + \frac{q}{\lambda} \, \cosh(\vartheta) \, \sinh(\vartheta) \Big) \Big] \, .
\eez
For real $q$ we obtain a second order Kuznetsov-Ma breather if $\lambda$ is real, and a 
second order Akhmediev breather if $\lambda$ is imaginary. 
This solution first appeared in \cite{AEK85}. 
See Fig.~\ref{fig:n=2_breathers} for corresponding plots.
\begin{figure} 
\begin{center}
\includegraphics[scale=.5]{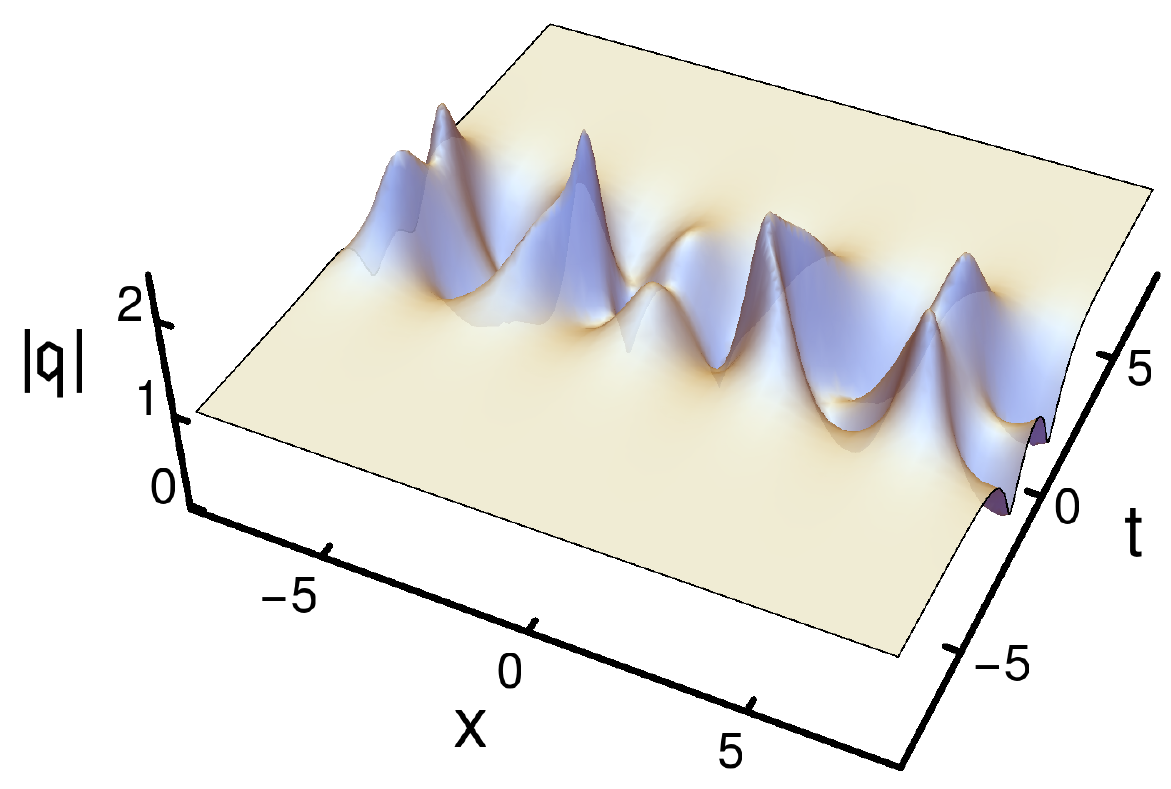}
\hspace{.5cm}
\includegraphics[scale=.5]{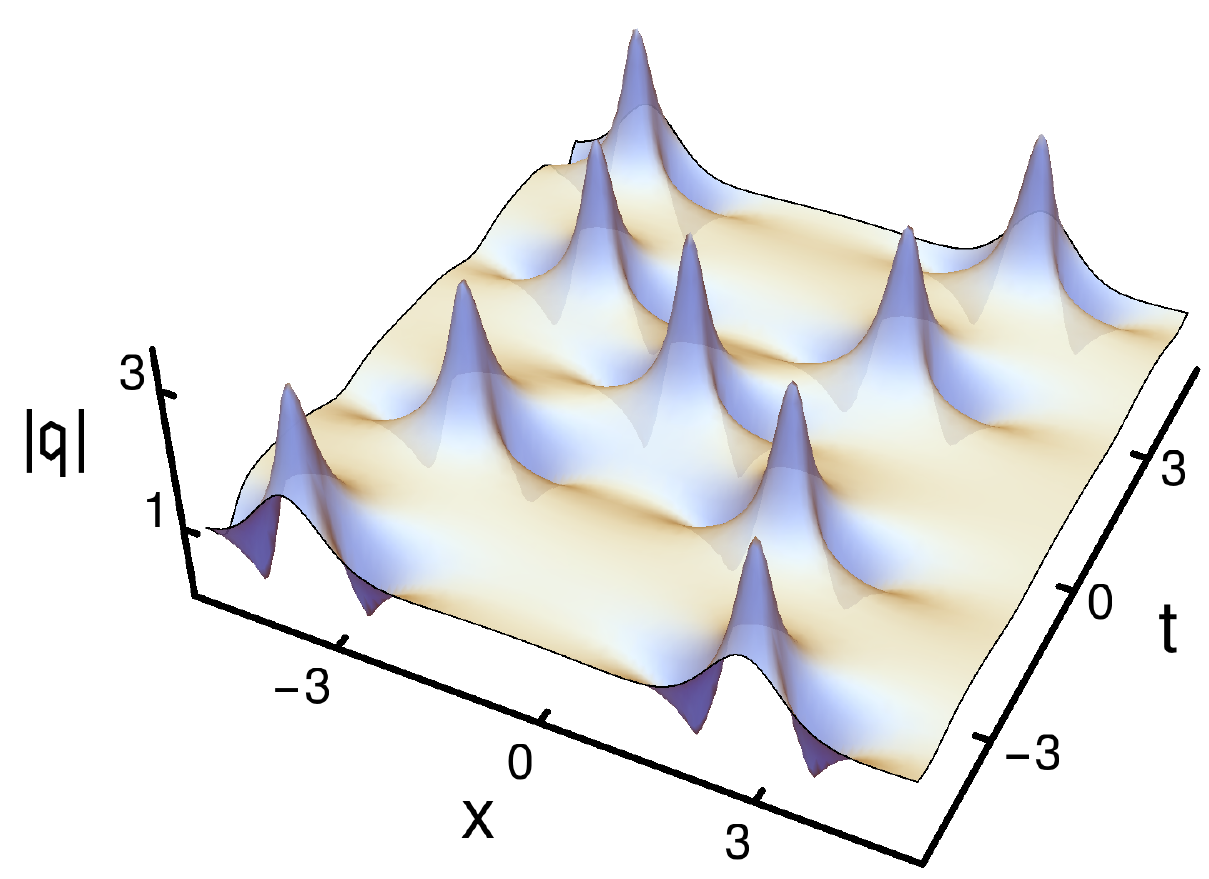}
\end{center}
\caption{Plots of the modulus of the second order ($n=2$) AKM solution in Example~\ref{ex:n=2breather},
with $q=1$ (a second order Akhmediev breather), 
respectively $q=9/4$ (a second order Kuznetsov-Ma breather). Here we set $c_1 = c_2 =0$. 
}
\label{fig:n=2_breathers} 
\end{figure}
\hfill $\Box$
\end{example}

\begin{example}
For $n=3$ we have
\bez
  &&  Q = \left( \begin{array}{ccc} q & 0 & 0 \\ 
                                    1 & q & 0 \\
                                    0 & 1 & q
                \end{array} \right) \, ,    \qquad
    \Lambda = \left( \begin{array}{ccc} \lambda & 0 & 0 \\ 
                \frac{q}{4 \lambda} & \lambda & 0 \\
                - \frac{1}{8 \lambda^3} & \frac{q}{4 \lambda} & \lambda
                \end{array} \right) \, ,    \qquad
    \Theta = \left( \begin{array}{ccc} \vartheta & 0 & 0 \\ 
                \frac{1}{4 \lambda} \mathcal{P}_1 & \vartheta & 0 \\
                - \frac{1}{8 \lambda^3} \mathcal{P}_2 & \frac{1}{4 \lambda} \mathcal{P}_1 & \vartheta
                \end{array} \right) \, ,    \\
  &&  \eta_1 = \left( \begin{array}{c} \cosh(\vartheta) \\ 
                \frac{1}{4 \lambda} \mathcal{P}_1 \, \sinh(\vartheta)  \\
                \frac{1}{32 \lambda^3} ( \lambda \, \mathcal{P}_1^2 \, \cosh(\vartheta) 
                   - 4 \mathcal{P}_2 \, \sinh(\vartheta) ) 
                \end{array} \right) \, , \quad
     \eta_2 = \left( \begin{array}{c} \theta_1 \\ 
                \frac{1}{4 \lambda} (\mathcal{P}_1+2) \, \theta_2  \\
                \frac{1}{64 \lambda^3} \, \theta_3
                \end{array} \right)  \, ,
\eez
where $\vartheta, \theta_1, \theta_2$ are defined as in Example~\ref{ex:n=2breather} 
with $\mathcal{P}_1 = \mathcal{P}$, and
\bez
  && \mathcal{P}_2 = x - \imag \, (2 \lambda^2 - 1 ) \, q \, t - 8 \, \lambda^3 c_3  \, , \\
  && \theta_3 = 2 \lambda \, \mathcal{P}_1 (\mathcal{P}_1 + 4 ) \, \theta_1 
     - 8 \, \mathcal{P}_2 \, \theta_2 - 8 \, \sinh (\vartheta)  \, .
\eez
Inserting (\ref{invOm_n=2}) and the expressions in Example~\ref{ex:n=3Jordan_Lyapunov_sol} in (\ref{Omega_inverse}), 
via (\ref{NLS_q_sol}) we can compute the third order AKM breather. This results in a very lengthy expression and 
will therefore not be written explicitly. 
\hfill $\Box$
\end{example}

\subsection{The family of (higher) Peregrine breathers}
\label{subsec:Peregrine_breather}

\begin{example}
\label{ex:Peregrine_breather}
Let $n=1$ and $q=2$. Then the solution of the linear system (\ref{eta_lin_sys}) 
is given by $\eta_1 = a \, [ x - x_0 + 2 \imag \, (t-t_0)]$ and $\eta_2 = \eta_1 + a$, 
with real constants $x_0,t_0$, and a complex constant $a$. The 
corresponding solution of (\ref{NLS_Lyapunov_2}) is
\bez
    \Omega_{(1)} = \frac{1}{4} \left( |\eta_1|^2 + |\eta_2|^2 \right)
           = \frac{1}{8} |a|^2 \Big( 16 \, (t-t_0)^2 + 4 (x-x_0 + \frac{1}{2})^2 + 1 \Big) \, ,
\eez
which nowhere vanishes (if $a \neq 0$), and from (\ref{NLS_phi_sol}) we obtain the \emph{Peregrine breather} 
\cite{Pere83}
\bez
    \mathfrak{q} 
    = e^{2 \imag t} \left( 1 - \frac{4 (1 + 4 \imag (t-t_0)}{16 \, (t-t_0)^2 
          + 4 (x-x_0)^2 + 1} \right) \, ,
\eez
after a redefinition of $x_0$ and $\mathfrak{q} \mapsto -\mathfrak{q}$. 
\hfill  $\Box$
\end{example}

Now we address the case where the matrix 
\bez
     \mathcal{N} := \frac{1}{4} \, Q^2 - I
\eez
in the linear system (\ref{eta_lin_sys}) is degenerate. More precisely, we will concentrate 
on the case where $\mathcal{N}$ is nilpotent, so that zero is the only eigenvalue, which then means 
that the only eigenvalue of $Q$ is $q=2$. 

\begin{proposition}
\label{prop:linsys_P}
Let $n>0$. Let the $n \times n$ matrix $Q$ be invertible and such that $\mathcal{N} := \frac{1}{4} Q^2 -I$ is nilpotent of degree $n$, 
i.e., $\mathcal{N}^n = 0$. 
Then the linear system (\ref{eta_lin_sys}) is solved by
\bez
   \eta_1 &=& \left( Q^{-1} R_2(\mathcal{N},x) \, R_1(-Q^2 \mathcal{N},t)
              + \imag \, R_1(\mathcal{N},x) \, R_2(-Q^2 \mathcal{N},t) \right) \, a \\
  && + \left( R_1(\mathcal{N},x) \, R_1(-Q^2 \mathcal{N},t) + \imag \, Q \mathcal{N} 
            \, R_2(\mathcal{N},x) \, R_2(-Q^2 \mathcal{N},t) \right) \, b \, , \\
   \eta_2 &=& \frac{1}{2} Q \, \eta_1 + \eta_{1,x} \, ,
\eez
where $a$ and $b$ are $n$-component constant vectors, and 
\bez
     R_1(\mathcal{N},x) := \sum_{k=0}^{n-1} \frac{x^{2k}}{(2k)!} \, \mathcal{N}^k \, , \qquad
     R_2(\mathcal{N},x) := \sum_{k=0}^{n-1} \frac{x^{2k+1}}{(2k+1)!} \, \mathcal{N}^k \, .
\eez
We note that $R_1(\mathcal{N},x)_x = \mathcal{N} \, R_2(\mathcal{N},x)$ and $R_2(\mathcal{N},x)_x = R_1(\mathcal{N},x)$.
\end{proposition}
\begin{proof}
Expressing the first of equations (\ref{eta_lin_sys}) as a first order system, we obtain the solution
\bez
  \left( \begin{array}{c} \eta_1 \\ \eta_{1,x} \end{array} \right) 
       = \exp( S \, x ) \, \left( \begin{array}{c} g_1 \\ g_2 \end{array} \right) 
       = \left( \cosh( S \, x ) + \sinh( S \, x ) \right) \left( \begin{array}{c} g_1 \\ g_2 \end{array} \right)
            \, , \qquad
       S := \left( \begin{array}{cc} 0 & I \\ \mathcal{N} & 0 \end{array} \right) \, ,  
\eez
where $g_i$, $i=1,2$, are $n$-component vectors, only dependent on $t$. 
Since $\mathcal{N}$ is assumed to be nilpotent of degree $n$, using the Taylor series expansions of $\cosh$ and $\sinh$, 
this yields
\bez
     \eta_1 = R_1(\mathcal{N},x) \, g_1 + R_2(\mathcal{N},x) \, g_2 \, .  
\eez 
The second of (\ref{eta_lin_sys}) now leads to 
\bez
     g_{1,tt} + Q^2 \mathcal{N} \, g_1 = 0 \, , \qquad g_2 = - \imag \, Q^{-1} g_{1,t} \, .
\eez
The first equation is solved in the same way as we solved the first of equations (\ref{eta_lin_sys}): 
\bez
     g_1 = R_1(-Q^2 \mathcal{N},t) \, b + R_2(-Q^2 \mathcal{N},t) \, \imag \, a \, ,
\eez
with constant complex $n$-component vectors $a,b$. 
\end{proof}

Solving the Lyapunov equations (\ref{NLS_Lyapunov_3}),  
with the solutions of the linear systems given by Proposition~\ref{prop:linsys_P}, then $\mathfrak{q}$ given 
by (\ref{NLS_q_sol}) is a quasi-rational solution of the focusing NLS equation. 
The following result establishes the existence of an infinite set of \emph{regular} quasi-rational solutions.

\begin{proposition}
Any solution of the focusing NLS equation obtained via Proposition~\ref{prop:NLS}  
with a solution of the linear system given by Proposition~\ref{prop:linsys_P}, where $Q$ is the Jordan 
block (\ref{Q_Jordan_block}) with eigenvalue $q=2$, is regular if the first component of the vector 
$a$ is different from zero.
\end{proposition}
\begin{proof}
Since all matrices appearing in Proposition~\ref{prop:linsys_P} are lower triangular, we can easily 
read off the first components of $\eta_1$ and $\eta_2$:
\be
     \eta_{11} = b_1 + (\frac{1}{2} x + \imag \, t) \, a_1 \, , \qquad
     \eta_{21} = \eta_{11} + \frac{1}{2} a_1 \, ,    \label{Peregrine_eta11}
\ee
where $a_1$ ($b_1$) is the first component of $a$ ($b$). Evidently, $\eta_{11}$ 
and $\eta_{21}$ cannot vanish simultaneously if $a_1 \neq 0$. 
Now regularity follows from Proposition~\ref{prop:Lyapunov_gen} in Appendix~B.
\end{proof}

Multiplying $\eta_1$, and thus also $\eta_2$, by a non-zero complex constant, leaves (\ref{NLS_q_sol}) invariant. 
Therefore the first component of $a$ can be set to $1$. As a consequence of Remark~\ref{rem:invariance}, 
we can set 
\bez
     a = (1,0,\ldots,0) \, .
\eez
Moreover, the fact that the linear system for $\eta_1$ is autonomous allows us to eliminate the 
first component of $b$. Indeed, for any shift $x \mapsto x - x_0$, $t \mapsto t - t_0$, 
there must be a choice of the parameters that compensates it, since our solution is the general one. 
(\ref{Peregrine_eta11}) with $a_1=1$ shows that this requires $b_1 \mapsto b_1 + \frac{1}{2} x_0 + \imag \, t_0$. 
Hence, without restriction of generality we may assume that
\bez
   b = (0,b_2,\ldots,b_n) \, , 
\eez
with complex constants $b_i$.

\begin{example}
Let $n=2$ and $Q$ the $2 \times 2$ Jordan block with eigenvalue $2$. We have
\bez
    \eta_1 = \left( \begin{array}{c} \vartheta \\
                    \frac{2}{3} \, \vartheta^3  - \frac{1}{4} x + b_2
                    \end{array} \right) \, , \quad
    \eta_2 = \left( \begin{array}{c} \vartheta + \frac{1}{2}  \\
                    \frac{2}{3} \vartheta^3 + \vartheta^2 + \frac{1}{2} \, \vartheta 
                     - \frac{1}{4} x + b_2 - \frac{1}{4}   
                    \end{array} \right)  \, ,     
\eez
where $\vartheta = \frac{1}{2} x + \imag \, t$. We obtain a quasi-rational solution of the focusing NLS equation. 
Replacing $x$ by $x-\frac{1}{2}$ and writing 
$b_2 = \beta + \frac{1}{6} - \frac{1}{8} \, \imag \, \alpha$, with real constants $\alpha$ and $\beta$, 
it can be written as
\bez
    \mathfrak{q} = 
    \left( 1 - \frac{G + \imag \, H}{D} \right) \, e^{2\,\imag\,t} \, , 
\eez    
where
\bez
  && G = \frac{1}{3} \, x^4 + \frac{1}{2} \, g_2(t) \, x^2 + 4 \, \beta \, x
         + g_0(t)  \, , \qquad 
     H =  \frac{4}{3} \, t \, x^4 + h_2(t) \, x^2 + 16 \, \beta \, t \, x
          + h_0(t)   \, ,  \\
  && D = 256 \, \det(\Omega_{(2)}) = \frac{1}{9} \, x^6 + \frac{1}{12} \, g_2(t) \, x^4
          - \frac{4}{3} \, \beta \, x^3 + f_2(t) \, x^2 + \beta \, g_2(t) \, x + f_0(t)  \, ,  \\
  && f_0(t) = \frac{64}{9} \, t^6 + 12 \, t^4 - \frac{4}{3} \, \alpha \, t^3 +\frac{11}{4}\, t^2
              - \frac{3}{4} \, \alpha \, t  + \frac{1}{16} \, \alpha^2 + 4 \, \beta^2 + \frac{1}{64} \, , \\
  && g_0(t) = \frac{80}{3} \, t^4 + 6 \, t^2 - \alpha \, t -\frac{1}{16} \, , \qquad 
     h_0(t) = \frac{64}{3} \, t^5 + \frac{8}{3} \, t^3 - 2 \, \alpha \, t^2 -\frac{5}{4} \, t
              + \frac{1}{8} \, \alpha \, , \\ 
  && f_2(t) = \frac{16}{3} \, t^4 - 2 \, t^2 +\alpha \, t + \frac{3}{16} \, , \quad
     g_2(t) = 16 \, t^2 + 1 \, , \quad 
     h_2(t) = \frac{32}{3} \, t^3 - 2 \, t + \frac{1}{2} \, \alpha  \, .
\eez
Up to differences in notation and scalings of $\alpha$ and $\beta$,  this is the second order Peregrine 
breather in \cite{Duba+Matv11}. It contains the second order Peregrine solution in \cite{AAT09,AAS09} 
as the special case $\alpha = \beta =0$. 
\hfill $\Box$
\end{example}

Fig.~\ref{fig:2nd_Peregrine} shows plots of the modulus of $\mathfrak{q}$ for the first few members 
of the infinite family of higher Peregrine breathers. Using computer algebra, expressions for the 
higher Peregrine breathers up to the tenth order have already been obtained (see \cite{Gail+Gast15,Gail15}).
The role played by the parameters of the higher Peregrine breathers has been clarified in \cite{AKA11,KAA11,KAA13}. 

\begin{figure} 
\begin{center}
\includegraphics[scale=.4]{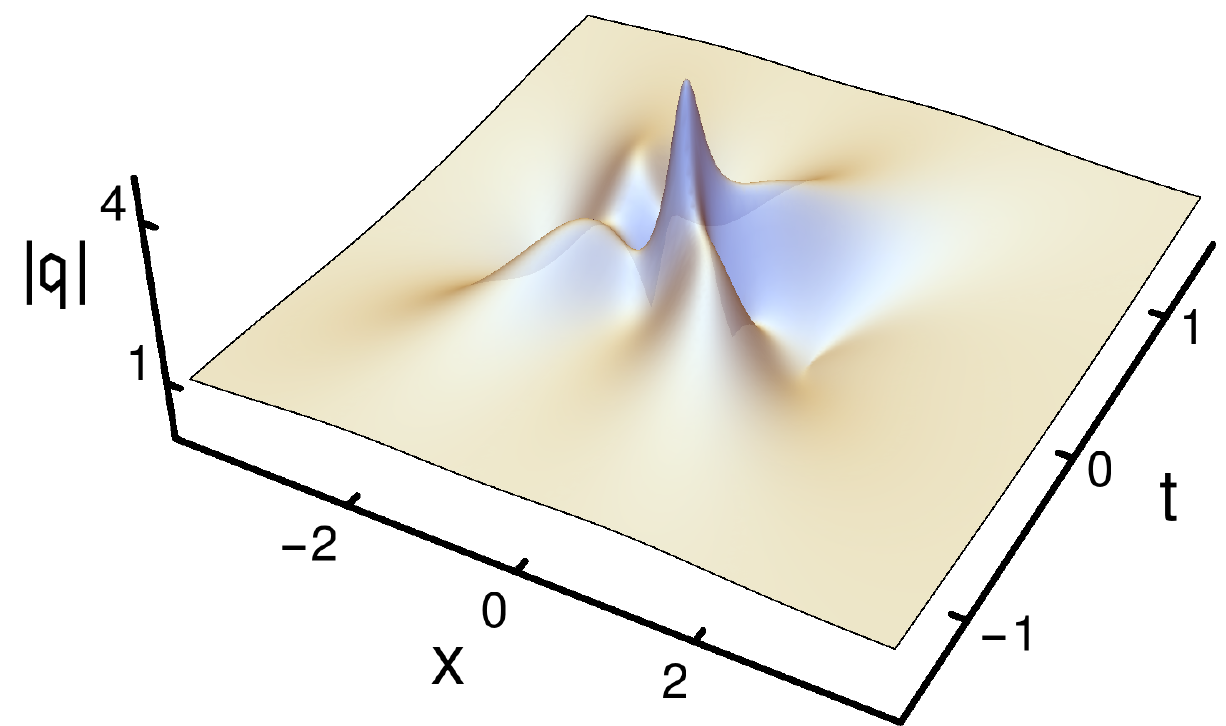} 
\hspace{.3cm}
\includegraphics[scale=.4]{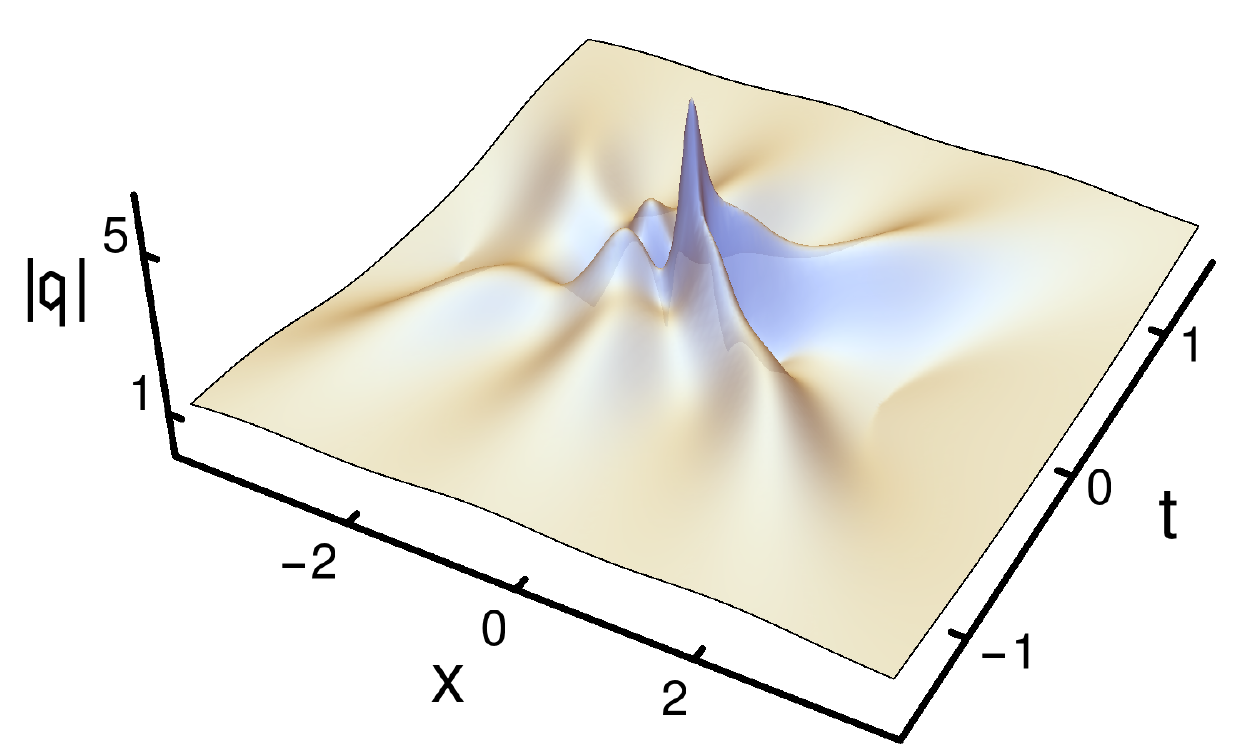}
\hspace{.3cm}
\includegraphics[scale=.4]{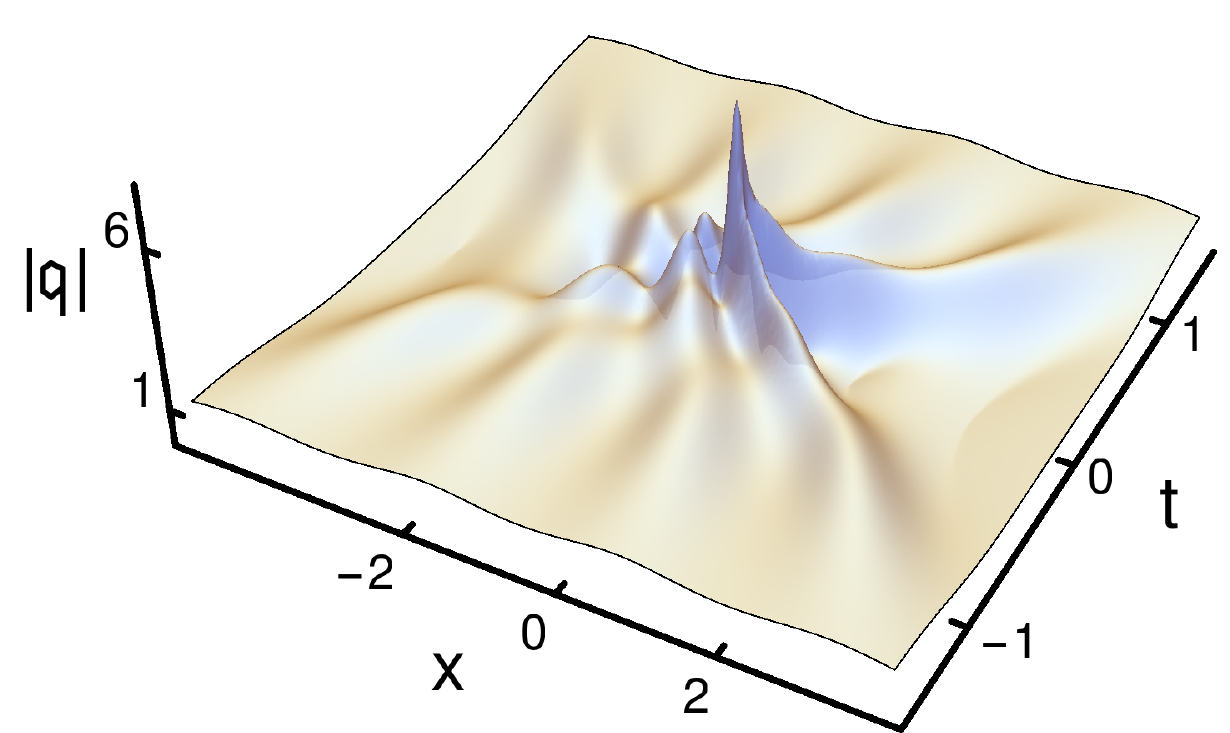}
\end{center}
\caption{Plots of the modulus of the second, third and fourth order Peregrine breather solution. Here 
we set $b = 0$. }
\label{fig:2nd_Peregrine} 
\end{figure}

\subsection{Nonlinear superpositions} 
If $Q = \mbox{block-diagonal}(Q_1, \ldots , Q_M)$, 
then the linear system (\ref{eta_lin_sys}) splits into the corresponding equations for the 
blocks $\{ Q_k \}$,
\bez
     \eta^{[k]}_{1,xx} - ( \frac{1}{4} Q_k^2 - I_k ) \, \eta^{[k]}_1 = 0 \, , \quad
     \imag \, \eta^{[k]}_{1,t} + Q_k \, \eta^{[k]}_{1,x} = 0 \, , \quad
     \eta^{[k]}_2 = \frac{1}{2} Q_k \, \eta^{[k]}_1 + \eta^{[k]}_{1,x}  \qquad k=1,\ldots,M \, .
\eez
Given solutions\footnote{In all but one of the $\eta^{[k]}$, representing (higher) breathers or rogue waves, 
we shall now allow $a_1$ to be different from $1$.} 
$\eta^{[k]}$, $k=1,\ldots,M$, 
all what remains to be done is to solve the Lyapunov equation with the above $Q$, and 
then to compute (\ref{NLS_q_sol}). An $n$th order AKM or Peregrine breather can be obtained as a special limit of such 
a (nonlinear) superposition of $n$ single breathers. But, of course, one can also superpose 
different (higher order) breathers. 
Fig.~\ref{fig:A2+P2} shows a superposition of a second order 
Akhmediev breather and a second order Peregrine breather. Also see \cite{Taji+Wata98} for other examples.

\begin{figure} 
\begin{center}
\includegraphics[scale=.4]{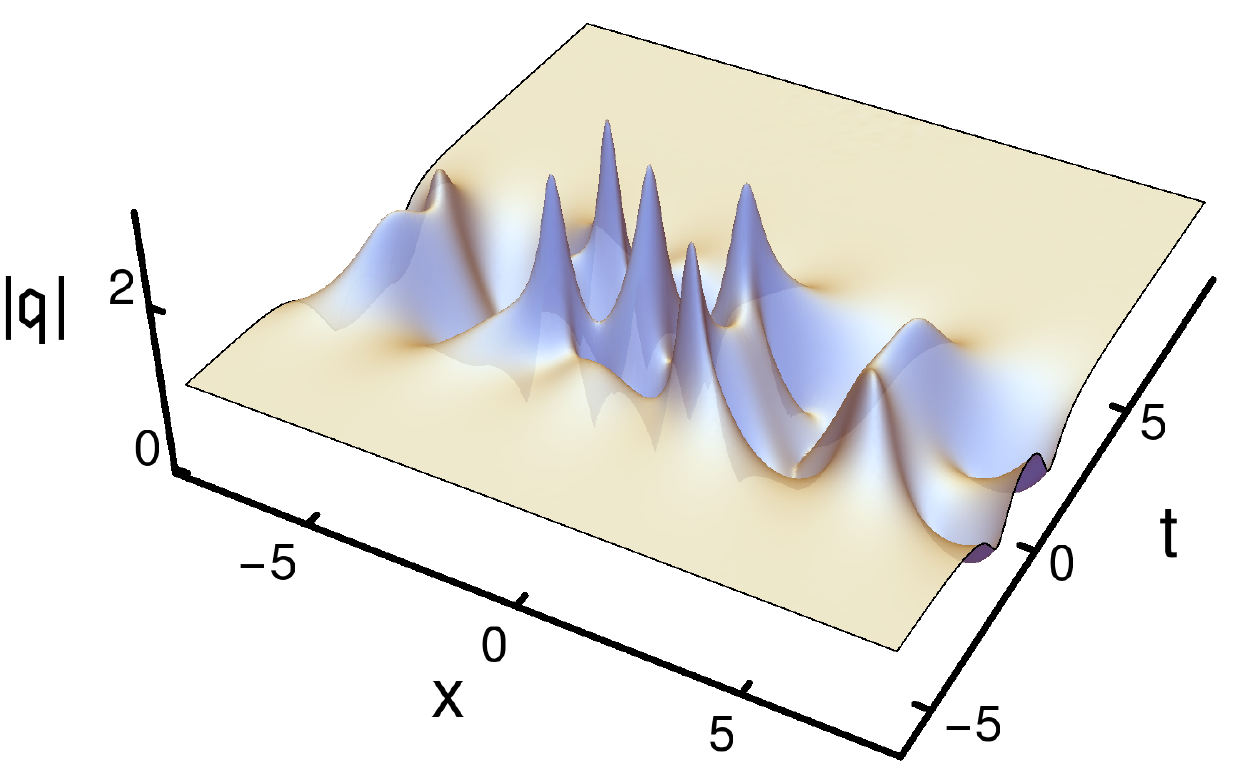} 
\hspace{.3cm}
\includegraphics[scale=.4]{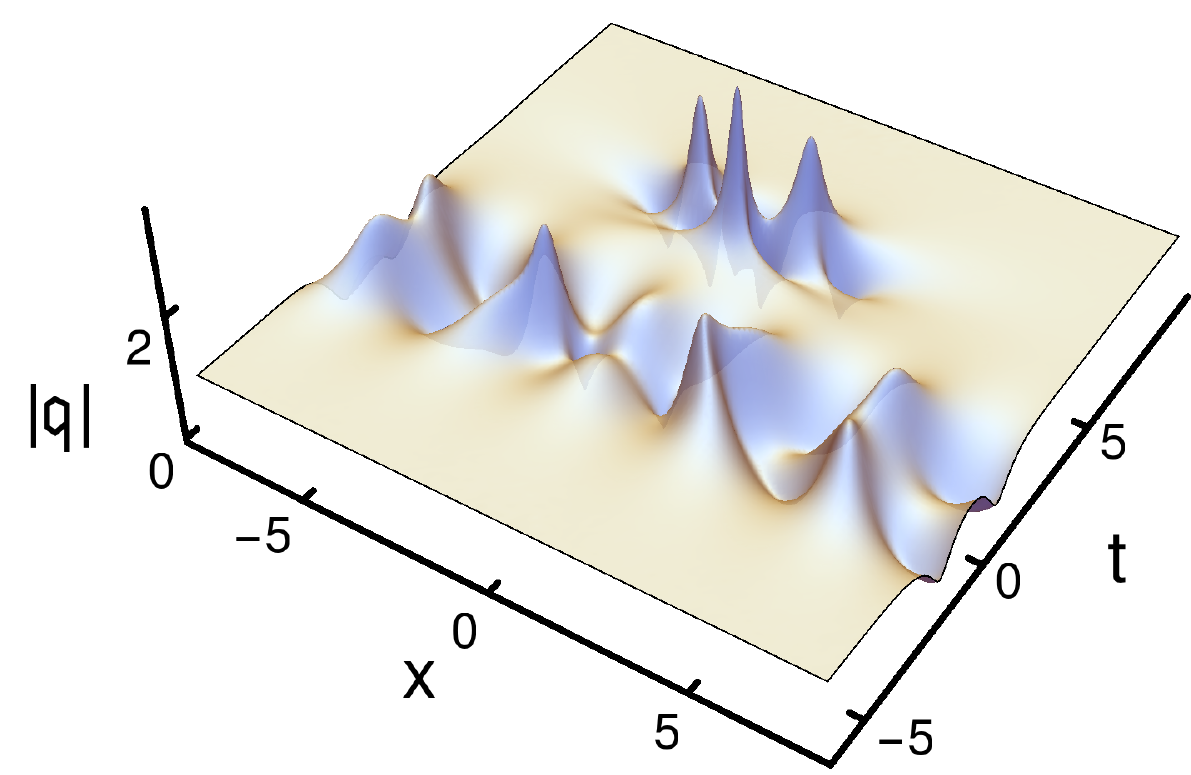} 
\end{center}
\caption{Plots of the modulus of $\mathfrak{q}$ for NLS solutions that are (nonlinear) superpositions 
of a second order Akhmediev breather (with $q=1$, $a_1=1$ and $c_1=c_2=0$) and a second order Peregrine breather 
(with $a_1=1$, $b_1=b_2=0$ and $t \mapsto t_0$). In the left plot we chose $t_0=0$, in the right $t_0=4$. In the 
latter case the Akhmediev breather appears at $t=0$ and (interpreting $t$ as time) essentially decays before 
the Peregrine breather shows up and then also decays. }
\label{fig:A2+P2} 
\end{figure}

\section{Conclusions}
\label{sec:concl}
In this work we presented a new derivation of the families of higher order breathers and rogue waves 
of the focusing NLS equation, using a vectorial Darboux transformation. The latter involves a rank two 
Lyapunov equation where the matrix $Q$ entering its left hand side is a Jordan block. The structure of 
the breathers and rogue waves (including the higher order Peregrine breathers) is then essentially determined 
by the corresponding solution of the Lyapunov equation. The different expressions for 
the higher AKM and Peregrine breathers, presented in this work, should be helpful for further explorations 
of the structure and properties of the members of this important family of exact solutions 
of the focusing NLS equation. 

The method presented in this work can easily be extended to vector and matrix versions of the NLS equation
(also see \cite{BDCW12,VSL13,Dega+Lomb13,LGZ14,MQG15} for different methods), its hierarchy \cite{AKCBA16}, 
self-consistent source extensions of the NLS equation (see \cite{CDMH16} and references cited there), 
and some related integrable equations. This will be addressed in forthcoming work.

\vspace{.3cm}
\noindent
\textbf{Acknowledgment}. The authors would like to thank Vladimir B. Matveev for some very motivating discussions 
about rogue waves.   
O.C. has been supported by an Alexander von Humboldt fellowship for postdoctoral researchers.

\renewcommand{\theequation} {\Alph{section}.\arabic{equation}}
\renewcommand{\thesection} {\Alph{section}}

\makeatletter
\newcommand\appendix@section[1]{
  \refstepcounter{section}
  \orig@section*{Appendix \@Alph\c@section: #1}
  \addcontentsline{toc}{section}{Appendix \@Alph\c@section: #1}
}
\let\orig@section\section
\g@addto@macro\appendix{\let\section\appendix@section}
\makeatother

\begin{appendix}

\section{Solution of the rank one Lyapunov equation with a Jordan block spectral matrix}
\label{app:Lyapunov_rank_one}
In this appendix we consider the rank one Lyapunov equation (see (\ref{NLS_Lyapunov_3}), or 
(\ref{Lyapunov_rank_one})) with $Q$ an $n \times n$ 
lower triangular Jordan block (\ref{Q_Jordan_block}), with eigenvalue $q$. 
If $\Re(q) \neq 0$, there is a unique solution and, as a consequence, $X$ is 
then Hermitian, i.e., $X^\dagger = X$. The following result is certainly known (also see, e.g., \cite{Zhang+Zhao13,XZZ14}, in 
the context of integrable systems), 
though we did not find a convenient reference. 

\begin{proposition}
\label{prop:rank_one}
Let $a = (a_1,\ldots,a_n)^\intercal$, with $a_i \in \bbC$, and
$\kappa := 2 \, \mathrm{Re}(q) \neq 0$. Then the Lyapunov equation 
\bez
       Q X + X Q^\dagger = a a^\dagger \, ,   
\eez
with the Jordan block (\ref{Q_Jordan_block}), has the unique solution
\be
      X = \kappa \, A \, K \, D \, S \, D \, K \, A^\dagger \, ,    \label{prop_X}
\ee
where
\be
      A = \left( \begin{array}{ccccc} a_1 & 0      & \cdots & \cdots & 0      \\ 
                                      a_2 & a_1    & 0      & \cdots & 0      \\ 
                                   \vdots & \ddots & \ddots & \ddots & \vdots \\
                                   \vdots & \ddots & \ddots & \ddots & 0      \\
                                      a_n & \cdots & \cdots & a_2 & a_1 
                 \end{array} \right) \, , \qquad
     \begin{array}{l} 
         D = \mathrm{diag}(1,-1,1,\ldots,-(-1)^n) \, ,  \\ \\
         K = \mathrm{diag}(\kappa^{-1},\kappa^{-2},\kappa^{-3},\ldots,\kappa^{-n}) \, , 
     \end{array}    
                                   \label{prop_A,D}
\ee
and $S$ is the $n \times n$ symmetric Pascal matrix
\be
     S = (s_{i,j}) \, , \qquad s_{i,j} = \frac{(i+j-2)!}{(i-1)! \, (j-1)!} = {i+j-2 \choose i-1} \, .  
              \label{prop_Pascal}
\ee
The determinant of $X$ is given by
\be
     \det(X) = \kappa^{-n^2} \, | a_1 |^{2n} \, .    \label{detX_m=1}
\ee
                       \hfill $\Box$
\end{proposition}
\begin{proof}
The lower triangular Toeplitz matrix $A$ in (\ref{prop_A,D})
commutes with the $n \times n$ Jordan block matrix $Q$ and achieves that
$a = A \, ( 1 , 0, \ldots, 0 )^\intercal$ (also see Lemma~6.5 in \cite{DMH10NLS}). 
Given a solution to the special Lyapunov equation
\bez
  Q \, Y + Y \, Q^\dagger = \left( \begin{array}{cccc} 1 & 0 & \cdots & 0 \\ 
                                           0 & 0 & \cdots & 0 \\
                                          \vdots & \vdots & \vdots & \vdots \\ 
                                          0 & \cdots & \cdots & 0 \end{array} \right) \, ,
\eez
a solution of (\ref{Lyapunov_rank_one}) is obtained as $X = A \, Y \, A^\dagger$. Assuming $\kappa \neq 0$, 
there are unique solutions $X$, respectively $Y$. The last relation implies 
\bez
      \det(X) = \det(Y) \, |\det(A)|^2 = \det(Y) \, |a_1|^{2n} \, .
\eez
It remains to compute $Y$ and then $\det(Y)$. The Lyapunov equation for $Y$ leads to  
linear recurrence relations. 
Writing
\bez
      Y_{i,j} = (-1)^{i+j} \, \kappa^{1-i-j} \, s_{i,j}  \, ,
\eez
we find that the coefficients $s_{i,j}$ are symmetric ($s_{i,j} = s_{j,i}$) and determined by 
\bez
       s_{1,1} = 1 \, , \quad 
       s_{1,j} = s_{1,j-1} \, , \quad  
       s_{i,j} = s_{i-1,j} + s_{i,j-1}  \qquad  i,j=2,\ldots,n \,  \, .
\eez
The solution is given by 
\bez
     s_{i,j} = \frac{(i+j-2)!}{(i-1)! \, (j-1)!} = {i+j-2 \choose i-1} 
               \qquad \quad  i,j=1,\ldots,n \, .
\eez
The matrix $S = (s_{i,j})$ is the $n \times n$ truncation of the infinite symmetric Pascal matrix 
(see, e.g., \cite{Edel+Strang04,Yates14}). 
Since
\bez
     Y_{i,j} = \kappa \, (-\kappa)^{-i} \, s_{i,j} \, (-\kappa)^{-j} \, ,
\eez
we have $Y = \kappa \, K \, D \, S \, D \, K$, and thus 
\bez
     \det(Y) = \kappa^{-n^2} \, \det(S) = \kappa^{-n^2} \, ,
\eez
since $\det(S)=1$ as an immediate consequence of the Cholesky decomposition (see, e.g., \cite{Edel+Strang04,Yates14}) 
\be
       S = L \, L^\intercal \, ,   \label{S_fact}
\ee
with the lower triangular Pascal matrix 
\bez
      L = (\ell_{i,j}) \, , \quad
      \ell_{i,j} = \left\{ \begin{array}{c@{\quad}l} {i-1 \choose j-1} & \mbox{if } j \geq i \\ 
                                                     0 & \mbox{otherwise} \end{array} \right. \, .
\eez
This completes the proof of Proposition~\ref{prop:rank_one}. 
\end{proof}

In particular, 
if $a_1 \neq 0$, the solution 
of the Lyapunov equation is invertible.\footnote{It is known (see \cite{Hearon77}, for example) that 
if $a_1 \neq 0$, then $(Q,a)$ is \emph{controllable} and the solution of the Lyapunov equation is invertible. }

\section{A result about the general Lyapunov equation}
In this appendix, we consider the general Lyapunov equation
\be
     Q \, X + X \, Q^\dagger = W \, ,         \label{Lyapunov_eq_W}
\ee
where $W$ is a Hermitian $n \times n$ matrix. If $Q$ and 
its negative Hermitian adjoint $-Q^\dagger$ satisfy the spectrum condition (\ref{spec}), 
i.e., if they have no eigenvalue in common, there is a unique solution $X$, for any Hermitian $W$.
As a consequence, $X$ is then Hermitian, i.e., $X^\dagger = X$. 
The matrix $W$ can be expressed as a sum of rank one matrices (dyadic products),
\be
     W = \sum_{k=1}^m V_k \, V_k^\dagger \, ,     \label{W_decomp}
\ee
with linearly independent vectors $V_k$. 
Assuming the spectrum condition, we have
\bez
       X = \sum_{k=1}^m X_k  \, ,
\eez
where $X_k$ solves the rank one Lyapunov equation 
\bez
      Q \, X_k + X_k \, Q^\dagger = V_k \, V_k^\dagger \, .   
\eez    

\begin{proposition}
\label{prop:Lyapunov_gen}
Let $Q$ be a lower triangular $n \times n$ Jordan block with eigenvalue $q$ and $\Re(q) \neq 0$. 
If the first component of one of the vectors $\{ V_k \}$ is 
different from zero, then the solution of (\ref{Lyapunov_eq_W}) is invertible. 
\end{proposition}
\begin{proof} 
Without restriction, let the first component $a_1$ of $a := V_1$ be non-zero. Then $X_1$ is given by 
(\ref{prop_X}). Since a symmetric Pascal matrix is positive definite \cite{Bhat07}, (\ref{prop_X}) 
is Hermitian and positive definite if $a_1 \neq 0$, and thus possesses a positive definite Hermitian square root. 
Then
\bez
    \det(X) = \det(X_1) \, \det\left( I + X_1^{-1/2} \, \Big(\sum_{k=2}^m X_k\Big) \, X_1^{-1/2} \right) \neq 0 \, .
\eez
In the last step we used the fact that, according to (\ref{prop_X}) and (\ref{S_fact}), each $X_k$ has a 
decomposition $X_k = B_k^\dagger \, B_k$, with a matrix $B_k$, so that $X_1^{-1/2} \, \sum_{k=2}^m X_k \, X_1^{-1/2}$ 
is positive semi-definite. 
\end{proof}

\end{appendix}


\small

\begin{thebibliography}{10}


\bibitem{Pere83}
Peregrine D 1983 Water waves, nonlinear {S}chr\"odinger equations and their
  solutions {\em J. Austral. Math. Soc. Ser. B\/} {\bf 25} 16--43

\bibitem{AAS09}
Akhmediev N, Ankiewicz A and Soto-Crespo J 2009 Rogue waves and rational
  solutions of the nonlinear {S}chr\"odinger equation {\em Phys. Rev. E\/} {\bf
  80} 026601

\bibitem{KFFMDGAD10}
Kibler B, Fatome J, Finot C, Millot G, Dias F, Genty G, Akhmediev N and Dudley
  J 2010 The {P}eregrine soliton in nonlinear fibre optics {\em Nature
  Physics\/} {\bf 6} 790--795

\bibitem{CHA11}
Chabchoub A, Hoffmann N and Akhmediev N 2011 Rogue wave observation in a water
  wave tank {\em Phys. Rev. Lett.\/} {\bf 106} 204502

\bibitem{BSN11}
Bailung H, Sharma S and Nakamura Y 2011 Observation of {P}eregrine solitons in
  a multicomponent plasma with negative ions {\em Phys. Rev. Lett.\/} {\bf 107}
  255005

\bibitem{BKA09}
Bludov Y, Konotop V and Akhmediev N 2009 Matter rogue waves {\em Phys. Rev.
  A\/} {\bf 80} 033610

\bibitem{Shri+Geog10}
Shrira V and Geogjaev V 2010 What makes the {P}eregrine soliton so special as a
  prototype of freak waves {\em J. Eng. Math.\/} {\bf 67} 11--22

\bibitem{AAT09}
Akhmediev N, Ankiewicz A and Taki M 2009 Waves that appear from nowhere and
  disappear without a trace {\em Phys. Lett. A\/} {\bf 373} 675--678

\bibitem{CHOA12}
Chabchoub A, Hoffmann N, Onorato M and Akhmediev N 2012 Super rogue waves:
  {O}bservation of a higher-order breather in water waves {\em Phys. Rev. X\/}
  {\bf 2} 011015

\bibitem{Matv+Sall91}
Matveev V and Salle M 1991 {\em Darboux {T}ransformations and {S}olitons\/}
  Springer Series in Nonlinear Dynamics (Berlin: Springer)

\bibitem{Matv00}
Matveev V 2000 {\em L.D. Faddeev's Seminar on Mathematical Physics\/} ({\em
  Advances in the Mathematical Sciences\/} vol 201) ed Semenov-Tian-Shansky M
  (Providence, R.I.: AMS) pp 179--209

\bibitem{GLL12}
Guo B, Ling L and Liu Q 2012 Nonlinear {S}chr\"odinger equation: {G}eneralized
  {D}arboux transformation and rogue wave solutions {\em Phys. Rev. E\/} {\bf
  85} 026607

\bibitem{DGKM10}
Dubard P, Gaillard P, Klein C and Matveev V 2010 On multi-rogue wave solutions
  of the {NLS} equation and positon solutions of the {KdV} equation {\em Eur.
  Phys. J. Special Topics\/} {\bf 185} 247--258

\bibitem{Duba+Matv11}
Dubard P and Matveev V 2011 Multi-rogue waves solutions of the focusing {NLS}
  equation and the {KP-I} equation {\em Nat. Hazards Earth Syst. Sci.\/} {\bf
  11} 667--672

\bibitem{Gail13}
Gaillard P 2013 Degenerate determinant representation of solutions of the
  nonlinear {S}chr\"odinger equation, higher order {P}eregrine breathers and
  multi-rogue waves {\em J. Math. Phys.\/} {\bf 54} 013504

\bibitem{Ohta+Yang12}
Ohta Y and Yang J 2012 General high-order rogue waves and their dynamics in the
  nonlinear {S}chr\"odinger equation {\em Proc. R. Soc. A\/} {\bf 468} 1716-1740

\bibitem{Akhm+Korn87}
Akhmediev N and Korneev V 1987 Modulation instability and periodic solutions of
  the {N}onlinear {S}chr\"odinger equation {\em Theor. Math. Phys.\/} {\bf 69}
  1089--1093

\bibitem{Kuzn77}
Kuznetsov E 1977 Solitons in a parametrically unstable plasma {\em Sov. Phys.
  Dokl.\/} {\bf 22} 507--508

\bibitem{Ma79}
Ma Y~C 1979 The perturbed plane-wave solutions of the cubic {S}chr\"odinger
  equation {\em Stud. Appl. Math.\/} {\bf 60} 43--58

\bibitem{KFFMGWADD12}
Kibler B, Fatome J, Finot C, Millot G, Genty G, Wetzel B, Akhmediev N, Dias F
  and Dudley J 2012 Observation of {K}uznetsov-{M}a soliton dynamics in optical
  fibre {\em Scientific Reports\/} {\bf 2} 463

\bibitem{CKDA14}
Chabchoub A, Kibler B, Dudley J and Akhmediev N 2014 Hydrodynamics of periodic
  breathers {\em Phil. Trans. R. Soc. A\/} {\bf 372} 20140005

\bibitem{DMH10NLS}
Dimakis A and M\"uller-Hoissen F 2010 Solutions of matrix {NLS} systems and
  their discretizations: a unified treatment {\em Inverse Problems\/} {\bf 26}
  095007

\bibitem{DMH13SIGMA}
Dimakis A and M\"uller-Hoissen F 2013 Binary {D}arboux transformations in
  bidifferential calculus and integrable reductions of vacuum {E}instein
  equations {\em SIGMA\/} {\bf 9} 009

\bibitem{CDMH16}
Chvartatskyi O, Dimakis A and M\"uller-Hoissen F 2016 Self-consistent sources
  for integrable equations via deformations of binary {D}arboux transformations
  {\em Lett. Math. Phys.\/} {\bf 106} 1139--1179
  
\bibitem{Lanc+Tism85} Lancaster P and Tismenetsky M 1985 {\em The Theory of Matrices\/}
  (Orlando, USA: Academic Press)
  
\bibitem{Horn+John13} Horn R and Johnson C 2013 {\em Matrix Analysis\/}
  (Cambridge, UK: Cambridge University Press) 
  
\bibitem{Sakh86} Sakhnovich L 1986 Factorization problems and operator identities 
 {\em Russ. Math. Surv. \/} {\bf 41} 1--64
  
\bibitem{March88} Marchenko V 1988 {\em Nonlinear Equations and Operator Algebras\/} 
 (Dordrecht, NL: Reidel)
 
\bibitem{SSR13} Sakhnovich L and Sakhnovich A and Roitberg I 2013  
 {\em Inverse Problems and Nonlinear Evolution Equations \/} (Berlin: De Gruyter)
 
\bibitem{ABDM10} Aktosun T, Busse T, Demontis F and van der Mee C 2010 
 Exact solutions to the Nonlinear Schr\"odinger equation 
 {\em Operator Theory: Advances and Applications} Vol. 203 (Basel: Birkh\"auser) 1--12
 
\bibitem{XZZ14}
Xu D, Zhang D and Zhao S 2014 The {S}ylvester equation and integrable
  equations: {I}. {T}he {K}orteweg-de {V}ries system and sine-{G}ordon equation
  {\em J. Nonl. Math. Phys.\/} {\bf 21} 382--406
 
\bibitem{Schie10LAA} Schiebold C 2010 Cauchy-type determinants and integrable systems 
 {\em Lin. Alg. Appl.\/} {\bf 433} 447--475
 
\bibitem{NAH09} Nijhoff F, Atkinson J and Hietarinta J 2009 Soliton solutions for 
 ABS lattice equations: I. Cauchy matrix approach {\em J. Phys. A: Math. Theor.\/} {\bf 42} 404005
 
\bibitem{JOR01}
Johnson R, Okubo K and Reams R 2001 Uniqueness of matrix square roots and an
  application {\em Lin. Alg. Appl.\/} {\bf 323} 51--60

\bibitem{Hasa97}
Hasan M 1997 A power method for computing square roots of complex matrices {\em
  J. Math. Anal. Appl.\/} {\bf 213} 393--405

\bibitem{Cros+Lanc74}
Cross G and Lancaster P 1974 Square roots of complex matrices {\em Linear and
  Multilinear Algebra\/} {\bf 1} 289--293

\bibitem{AEK87}
Akhmediev N, Eleonskii V and Kulagin N 1987 Exact first-order solutions of the
  nonlinear {S}chr\"odinger equation {\em Theor. Math. Phys.\/} {\bf 72}
  809--818

\bibitem{Bion+Kova14}
Biondini G and Kova\v{c}i\v{c} G 2014 Inverse scattering transform for the
  focusing nonlinear {S}chr\"odinger equation with nonzero boundary conditions
  {\em J. Math. Phys.\/} {\bf 55} 031506

\bibitem{ASA09PRA}
Akhmediev N, Soto-Crespo J and Ankiewicz A 2009 How to excite a rogue wave {\em
  Phys. Rev. A\/} {\bf 80} 043818
  
\bibitem{Higham87} Higham N 1987 Computing real square roots of a real matrix 
 {\em Linear Alg. Appl.\/} {\bf 88-89} 405--430

\bibitem{High08}
Higham N 2008 {\em Functions of {M}atrices, {T}heory and {C}omputation\/}
  (Philadelphia, PA, USA: SIAM)

\bibitem{AEK85}
Akhmediev N, Eleonskii V and Kulagin N 1985 Generation of periodic trains of
  picosecond pulses in an optical fiber: exact solutions {\em Sov. Phys.
  JETP\/} {\bf 62} 894--899

\bibitem{Gail+Gast15}
Gaillard Pand~Gastineau M 2015 The {P}eregrine breather of order nine and its
  deformations with sixteen parameters solutions to the {NLS} equation {\em
  Phys. Lett. A\/} {\bf 379} 1309--1313

\bibitem{Gail15}
Gaillard P 2015 Tenth {P}eregrine breather solution to the {NLS} equation {\em
  Ann. Phys.\/} {\bf 355} 293--298

\bibitem{AKA11}
Ankiewicz A, Kedziora D and Akhmediev N 2011 Rogue wave triplets {\em Phys.
  Lett. A\/} {\bf 375} 2782--2785

\bibitem{KAA11}
Kedziora D, Ankiewicz A and Akhmediev N 2011 Circular rogue wave clusters {\em
  Phys. Rev. E\/} {\bf 84} 056611

\bibitem{KAA13}
Kedziora D, Ankiewicz A and Akhmediev N 2013 Classifying the hierarchy of
  nonlinear-{S}chr\"odinger-equation rogue-wave solutions {\em Phys. Rev. E\/}
  {\bf 88} 013207

\bibitem{Taji+Wata98}
Tajiri M and Watanabe Y 1998 Breather solutions to the focusing nonlinear
  {S}chr\"odinger equation {\em Phys. Rev. E\/} {\bf 57} 3519--3519

\bibitem{BDCW12}
Baronio F, Degasperis A, Conforti M and Wabnitz S 2012 Solutions of the vector
  nonlinear {S}chr\"odinger equations: {E}vidence for deterministic rogue waves
  {\em Phys. Rev. Lett.\/} {\bf 109} 044102

\bibitem{VSL13}
Vishnu~Priya N, Senthilvelan M and Lakshmanan M 2013 Akhmediev breathers, Ma
  solitons, and general breathers from rogue waves: A study in the Manakov
  system {\em Phys. Rev. E\/} {\bf 88} 022918

\bibitem{Dega+Lomb13}
Degasperis A and Lombardo S 2013 Rational solitons of wave resonant-interaction
  models {\em Phys. Rev. E\/} {\bf 88} 052914

\bibitem{LGZ14}
Ling L, Guo B and Zhao L~C 2014 High-order rogue waves in vector nonlinear
  {S}chr\"odinger equations {\em Phys. Rev. E\/} {\bf 89} 041201(R)

\bibitem{MQG15}
Mu G, Qin Z and Grimshaw R 2015 Dynamics of rogue waves on a multi-soliton
  background in a vector nonlinear {S}chr\"odinger equation {\em SIAM J. Appl. Math.\/} 
  {\bf 75} 1--20 

\bibitem{AKCBA16}
Ankiewicz A, Kedziora D, Chowdury A, Bandelow U and Akhmediev N 2016 Infinite
  hierarchy of nonlinear {S}chr\"odinger equations and their solutions {\em
  Phys. Rev. E\/} {\bf 93} 012206

\bibitem{Zhang+Zhao13}
Zhang D and Zhao S 2013 Solutions to {ABS} lattice equations via generalized
  {C}auchy matrix approach {\em Stud. Appl. Math.\/} {\bf 131} 72--103

\bibitem{Edel+Strang04}
Edelman A and Strang G 2004 Pascal matrices {\em American Math. Monthly\/} {\bf
  111} 189--197

\bibitem{Yates14}
Yates L 2014 Linear algebra of {P}ascal matrices {\em report (student project),
  Math. Dept., Georgia College, Milledgeville, GA, USA,
  \url{https://www.gcsu.edu/sites/files/page-assets/node-808/attachments/yates.pdf}\/}

\bibitem{Hearon77}
Hearan J 1977 Nonsingular solutions of {$TA-BT=C$} {\em Lin. Alg. Appl.\/} {\bf
  16} 57--63

\bibitem{Bhat07}
Bhatia R 2007 {\em Positive Definite Matrices\/} (Princeton, NJ: Princeton
  University Press)

\end{thebibliography}
\end{document}